\documentclass[lettersize,journal]{IEEEtran}
\usepackage{amsmath,amsfonts}
\usepackage{algorithmic}
\usepackage{algorithm}
\usepackage{array}
\usepackage{textcomp}
\usepackage{stfloats}
\usepackage{url}
\usepackage{verbatim}
\usepackage{cite}

\usepackage{tikz}
\usepackage{amsmath}

\usepackage{filecontents}

\usepackage{graphicx}
\usepackage{subfigure}
\usepackage{amsmath}
\usepackage{amsthm}
\usepackage{mathrsfs}
\usepackage[vlined,ruled,linesnumbered,algo2e]{algorithm2e}
\usepackage{xcolor}
\usepackage{colortbl}
\usepackage{multicol}
\usepackage{multirow}
\usepackage{diagbox}
\usepackage{bbm}
\usepackage{makecell}
\usepackage{enumitem}
\usepackage{booktabs}
\usepackage{algorithm}
\usepackage{algorithmic}
\usepackage{hyperref}
\usepackage{amsfonts}
\usepackage{newtxmath}
\usepackage{cite}
\usepackage[numbers]{natbib}

\newtheorem{thm}{Theorem}

\newtheorem{lem}{Lemma}

\begin{document}

\title{Secure Forgetting: A Framework for Privacy-Driven Unlearning \\ in Large Language Model (LLM)-Based Agents}

\author{Dayong Ye, Tianqing Zhu$^*$, Congcong Zhu, Feng He, Qi He, Wang Shang, Bo Liu, and Wanlei Zhou
\thanks{$^*$Tianqing Zhu is the corresponding author.}
}




\maketitle

\begin{abstract}
Large language model (LLM)-based agents have recently gained considerable attention due to the powerful reasoning capabilities of LLMs. Existing research predominantly focuses on enhancing the task performance of these agents in diverse scenarios. However, as LLM-based agents become increasingly integrated into real-world applications, significant concerns emerge regarding their accumulation of sensitive or outdated knowledge. Addressing these concerns requires the development of mechanisms that allow agents to selectively forget previously learned knowledge, giving rise to a new term \emph{LLM-based agent unlearning}.
This paper initiates research on unlearning in LLM-based agents. Specifically, we propose a novel and comprehensive framework that categorizes unlearning scenarios into three contexts: state unlearning (forgetting specific states or items), trajectory unlearning (forgetting sequences of actions) and environment unlearning (forgetting entire environments or categories of tasks). Within this framework, we introduce a natural language-based unlearning method that trains a conversion model to transform high-level unlearning requests into actionable unlearning prompts, guiding agents through a controlled forgetting process. Moreover, to evaluate the robustness of the proposed framework, we introduce an unlearning inference adversary capable of crafting prompts, querying agents, and observing their behaviors in an attempt to infer the forgotten knowledge. Experimental results show that our approach effectively enables agents to forget targeted knowledge while preserving performance on untargeted tasks, and prevents the adversary from inferring the forgotten knowledge. 
\end{abstract}



\section{Introduction}
Large Language Models (LLMs) have demonstrated their potential in various applications, from creative writing to code generation \cite{Wan24TMLR}. Among these, an important application is LLM-driven agents, widely used in AI-powered personal assistants, such as Google's Gemini \cite{GoogleGemini}.
A typical agent is an intelligent entity that perceives its environment and performs actions to maximize its payoffs \cite{Mnih15Nature}. Unlike deep neural network-based agents, which require parameter update during training, LLM-based agents operate by querying the LLM and executing actions based on its responses \cite{Wang24Survey}.
Current research on LLM-based agents focuses on enhancing response quality to improve agent decision-making \cite{Shinn23NIPS,Guo24IJCAI}. However, a crucial yet largely unexplored challenge remains: how to make an LLM-based agent forget previously acquired knowledge. We term this emerging problem as \emph{LLM-based agent unlearning}. 

Machine unlearning has been extensively studied in neural network–based models \cite{Bourtoule21, Xu23}, and recent work has expanded this line of research to LLMs \cite{Liu24NatureMI, Yao24ICLR, Liu24NIPS}.
However, LLM unlearning fundamentally differs from LLM-based agent unlearning. The key difference lies in their objectives: LLM unlearning aims to remove the influence of specific knowledge from the model itself, whereas LLM-based agent unlearning seeks to regulate the agent’s behavior without modifying the underlying LLM, as direct modification of LLM parameters is often restricted in practice \cite{Zhang24ICML}. 
Another branch of machine unlearning research focuses on reinforcement unlearning \cite{Ye25NDSS,Gong25NDSS}, which aims to selectively forget an agent's learned knowledge within a training environment. This approach, however, also differs from LLM-based agent unlearning. Reinforcement unlearning is achieved by directly modifying the model parameters to alter action selection probabilities of the agent. In contrast, LLM-based agents are not standalone neural networks. Their behavior is governed by LLMs, driven by user-provided prompts, without modification of internal model parameters. 
These distinctions highlight the need to establish a new research direction: LLM-based agent unlearning, which introduces two research challenges.


\vspace{1mm}
\noindent\textbf{How can we concretize the unlearning objectives in LLM-based agent unlearning?}
LLM unlearning is typically designed with specific unlearning targets, such as removing knowledge derived from a given dataset \cite{Yao24ICLR}. 
In contrast, LLM-based agent unlearning focuses on modifying the agent's behavior, such as intentionally degrading its performance in a given environment \cite{Ye25NDSS}, often without explicitly defined unlearning objectives.

\vspace{1mm}
\noindent\textbf{How can the performance of an LLM-based agent be maintained after unlearning?} Unlearning can lead to catastrophic forgetting, where the model's performance deteriorates not only on the unlearned tasks but also on the remaining ones. This issue is also critical in LLM-based agents, 
causing the agent to exhibit degraded performance in environments where it should still function optimally. 

\vspace{1mm}
To address the first challenge, we categorize unlearning requests into three distinct scenarios, each representing a specific unlearning objective. For the second challenge, we design unlearning requests in natural language and train a dedicated conversion model that translates these requests into actionable unlearning prompts, ensuring that the LLM effectively controls the agent’s behavior to achieve the desired unlearning outcomes while preserving its performance on unrelated tasks. In summary, we make \textbf{three contributions}.
\begin{itemize}[leftmargin=*]
    \item We take a significant step forward in investigating LLM-based agent unlearning. Our approach enables selective forgetting of learned knowledge by directly shaping the agent’s behavior, offering new conceptual insights for this emerging research area.
    \item We propose a novel framework that categorizes unlearning requests into three distinct scenarios. 
    Within this framework, we develop a natural language-based unlearning method that trains a conversion model to translate unlearning requests into actionable  prompts, enabling precise manipulation of agent behavior through LLMs. 
    \item We conduct both theoretical analysis and empirical evaluation of the proposed unlearning framework in conventional learning tasks as well as under adversarial settings, demonstrating both its effectiveness and robustness.
\end{itemize}

\section{Preliminary}
\noindent\textbf{Large Language Models (LLM).} 
A large language model, denoted as $LLM$, takes a text sample $x$, often referred to as a `prompt', as input and generates another text sample $y$, typically referred to as a `response', i.e., $y=LLM(x)$. 
Text samples processed by LLMs are represented as sequences of tokens: $x = [x^1, x^2, \cdots, x^n]$, where $n$ represents the number of tokens in $x$. These tokens are discrete units of language, which could be words or even characters, depending on the tokenization scheme used. 

\vspace{1mm}
\noindent\textbf{LLM-based Agents.} An LLM-based agent is developed to generate texts and actions contingent upon observed states. Specifically, an action $a_t$ is determined at time $t$ based on the current state $s_t$ and a policy $\pi$, implemented by an LLM. Upon executing $a_t$, the agent receives a scalar reward $r_t$, which evaluates the suitability of $a_t$ in the state $s_t$, subsequently transitioning to the next state $s_{t+1}$. Over time, as the agent executes a series of actions, $a_1,...,a_t$, a trajectory $\tau=\{(s_1,a_1),...,(s_t,a_t)\}$ is formed. 
All available actions $\mathcal{A}$, states $\mathcal{S}$, the transition function $\mathcal{T}:\mathcal{S}\times\mathcal{A}\rightarrow\mathcal{S}$, and the reward function $\mathcal{R}:\mathcal{S}\times\mathcal{A}\rightarrow\mathbb{R}$ collectively define the learning environment $\mathcal{E}=(\mathcal{S},\mathcal{A},\mathcal{T},\mathcal{R})$.

Unlike conventional agent training, which iteratively updates the policy $\pi$, an LLM-based agent typically enhances its performance by collecting and utilizing feedback from an LLM \cite{Shinn23NIPS}. This feedback, along with previously observed states and actions, is stored in the agent's memory, denoted as $Mem$, and is used as part of the input prompt to the LLM to guide the agent's future behavior. 
An overview of an LLM-based agent is shown in Figure \ref{fig:LLMAgent}. The agent observes a state from the environment, e.g., the requirements of a task, and takes an action, e.g., performing the task. The environment provides a reward feedback, reflecting how well the task was performed. The agent stores the state, action, and reward in its memory, typically implemented as a JSON file, and uses this accumulated information to improve its performance on subsequent execution of the task.

\begin{figure}[ht]
\centering
	\includegraphics[scale=0.35]{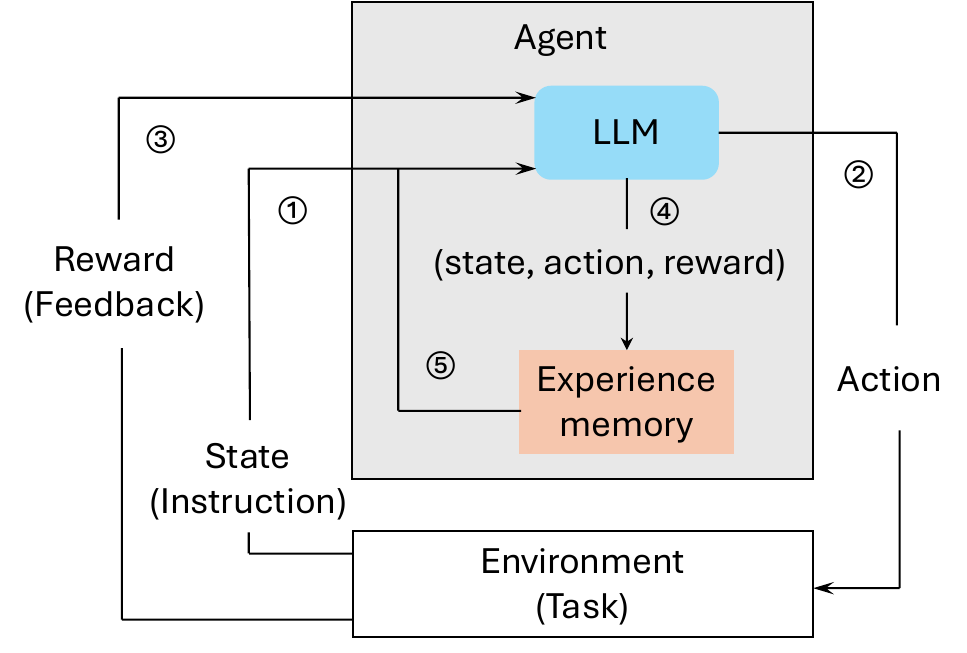}
	\caption{Overview of an LLM-based agent. The agent $\textcircled{1}$ observes a state from the environment and $\textcircled{2}$ takes an action in response. The environment $\textcircled{3}$ returns a reward reflecting how well the action was performed. The agent $\textcircled{4}$ stores the state, action, and reward in its memory, $\textcircled{5}$ using this accumulated information to improve future performance.}
    \vspace{-1mm}
	\label{fig:LLMAgent}
\end{figure}

The notations used in this paper are shown in Table \ref{tab:notation}.

\begin{table}[!ht]\scriptsize
	\centering
 \vspace{-1mm}
	\caption{Summary of Notations}
    \vspace{-0mm}
\begin{tabular} {cl}
\toprule
{\bf Notations} & \hspace{8em}{\bf Description} \\
\midrule
$\mathcal{S}$ & A state set, $\mathcal{S}=\{s_1,...,s_n\}$  \\
$\mathcal{A}$ & An action set, $\mathcal{A}=\{a_1,...,a_m\}$  \\
$\mathcal{T}$ & A transition function, $\mathcal{T}:\mathcal{S}\times\mathcal{A}\rightarrow\mathcal{S}$ \\
$\mathcal{R}$ & A reward function, $\mathcal{R}:\mathcal{S}\times\mathcal{A}\rightarrow\mathbb{R}$ \\
$\mathcal{E}$ & A learning environment/task, $\mathcal{E}=(\mathcal{S},\mathcal{A},\mathcal{T},\mathcal{R})$ \\
$\tau$ & \makecell[l]{A trajectory consisting of state-action pairs: \\\hspace{2mm} $\tau=\{(s_1,a_1),...,(s_t,a_t)\}$}  \\
$Mem$ & \makecell[l]{The agent's memory, $Mem = {(s, a, r)}$, \\\hspace{2mm} storing previously encountered states, actions, and rewards}\\
$\mathcal{X}$ & \makecell[l]{A prompt space, $\mathcal{X}=\{x_1,x_2,...\}$, \\\hspace{2mm} representing a set of diverse prompts}\\
$\pi$ & A policy learned by, or implemented using an LLM for, an agent \\
\bottomrule
\end{tabular}
	\label{tab:notation}
    \vspace{-1mm}
\end{table}

\section{Threat Model and Problem Definition}
We explore three unlearning scenarios, each aligned with a different problem definition. In all cases, the LLM-based agent is controlled by its owner, who can query the LLM, observe its outputs, and use those outputs to guide the agent’s behavior. However, the agent owner cannot access or modify the internal parameters of the LLM. Unlearning is initiated by the agent owner upon identifying that previously used training information has become sensitive or outdated. 


\vspace{1mm}
\noindent\textbf{Scenario 1: Unlearning of States Within a Trajectory.} This unlearning scenario focuses on selectively forgetting specific states, inspired by the agent's learned recognition of restricted areas. For example, consider a self-driving car agent powered by an LLM \cite{Cui24} that has been trained to navigate a city and, in the process, has learned the locations of private residences or government facilities. If the owner of one of these properties invokes their right to data erasure, the agent must unlearn not only the identity of the location but also its spatial relevance within the driving environment. 


\vspace{1mm}
\noindent\textbf{Scenario 2: Unlearning of Trajectories Within an Environment.} This unlearning scenario focuses on selectively erasing specific trajectories within an environment. Consider a domestic service robot equipped with an LLM‑based reasoning module that plans household tasks through trajectory‑level decision making \cite{Wang24CCC}. The robot may learn trajectories that encode private information. For instance, it might repeatedly navigate from the living room to a medication cabinet after hearing the owner mention ``pills'', exposing the owner’s medical condition. 
If the robot is remotely accessed by an adversary, such trajectories could enable inference of the owner’s health information. To mitigate this risk, the robot must be able to unlearn these sensitive trajectories upon request. 





\vspace{1mm}
\noindent\textbf{Scenario 3: Unlearning Across a Spectrum of Training Environments.} This unlearning scenario focuses on forgetting a learning environment. Consider an LLM-based  delivery drone \cite{Chen25TMC} trained across multiple urban and rural environments. During training, the drone learns distinct flight policies optimized for various regions, for example, a coastal industrial zone where metal structures and strong wind patterns influence its flight control strategy.
Later, this industrial zone is decommissioned and classified as a restricted airspace due to security concerns. Retaining learned knowledge of this environment poses a security risk, where an adversary gaining access to the LLM’s prompt history or responses could reconstruct sensitive information about the environment’s vulnerabilities. To mitigate this risk, the drone must selectively unlearn all knowledge derived from this environment. 




\vspace{1mm}
\noindent\textbf{Problem Definition.} 
Since the policy $\pi$ is implemented by an LLM and its internal parameters cannot be modified, we cannot enforce unlearning by altering states or transition probabilities. 
Instead, unlearning must be achieved through two complementary strategies: (1) modifying the agent’s memory $Mem$ to prevent it from relying on the targeted knowledge, and (2) controlling the agent’s behavior to ensure that forgotten information no longer influences its actions.

\begin{itemize}[leftmargin=*]
\item \textbf{Scenario 1.} Given a trajectory $\tau=\{(s_1,a_1),...,(s_m,a_m)\}$, let the set of unlearning states be $\mathcal{S}_u$, and the agent's memory be $Mem$, we aim to achieve:
\begin{equation}\label{eq:S1ClearMem}
    Mem'(s)=\emptyset, \forall s\in\mathcal{S}_u,
\end{equation}
\begin{equation}\label{eq:S1AvoidStates}
    \mathcal{T}(s_{i+1}|s_i,a_i;Mem')=0, \forall s_{i+1}\in\mathcal{S}_u,
\end{equation}
where $1\leq i<m$, and 
\begin{equation}\label{eq:S1KeepPerf}
    \mathcal{T}(s_{k+1}|s_k,a_k;Mem')=\mathcal{T}(s_{k+1}|s_k,a_k;Mem), \forall s_{k+1}\notin\mathcal{S}_u.
\end{equation}
Eq. \ref{eq:S1ClearMem} focuses on erasing the memory associated with the unlearning states. Eq. \ref{eq:S1AvoidStates} aims to prevent the agent from revisiting these states once the memory has been updated. Eq. \ref{eq:S1KeepPerf} seeks to maintain the agent's performance in the remaining states after unlearning. Note that simply removing the corresponding memory, without imposing behavioral constraints, is insufficient to achieve unlearning, as the agent may reacquire this memory in the future.

\item \textbf{Scenario 2.} Given a trajectory $\tau=\{(s_1,a_1),...,(s_m,a_m)\}$, to unlearn it, we aim to achieve
\begin{equation}\label{eq:S2ClearMem}
    Mem'(\tau)=\emptyset, 
\end{equation}
\begin{equation}\label{eq:S2AvoidTraj}
    \mathcal{T}(s_{i+1}|s_i,a_i;Mem')=0,
\end{equation}
where $1\leq i<m$, and
\begin{equation}\label{eq:S2KeepPerf}
    \mathcal{T}(s_{k+1}|s_k,a_k;Mem')=\mathcal{T}(s_{k+1}|s_k,a_k;Mem), \forall s_{k+1}\notin\tau.
\end{equation}
Eq. \ref{eq:S2ClearMem} focuses on removing the memory associated with the unlearning trajectory. Eq. \ref{eq:S2AvoidTraj} aims to prevent the formation of the unlearned trajectory in the future. Eq. \ref{eq:S2KeepPerf} ensures that the agent maintains performance on all other trajectories. 
The key difference between Scenarios 1 and 2 lies in the restrictions placed on the agent's access to states. In Scenario 1, the objective is to completely prohibit the agent from accessing the forgotten states. In Scenario 2, the agent is permitted to visit the states in the unlearning trajectory, but it must not traverse them in the predetermined sequence.

\item \textbf{Scenario 3.} Given a training environment $\mathcal{E}=(\mathcal{S},\mathcal{A},\mathcal{T},\mathcal{R})$, the objectives for unlearning it are:
\begin{equation}\label{eq:S3ClearMem}
    Mem'(s)=\emptyset, \forall s\in\mathcal{S},
\end{equation}
\begin{equation}\label{eq:S3ReducePerf}
   \mathcal{R}(s,\pi(s;Mem'))<\mathcal{R}(s,\pi(s;Mem)), \forall s\in\mathcal{S},
\end{equation}
\begin{equation}\label{eq:S3KeepPerf}
   \mathcal{R}(s,\pi(s;Mem'))=\mathcal{R}(s,\pi(s;Mem)), \forall s\notin\mathcal{S}.
\end{equation}
Eq. \ref{eq:S3ClearMem} aims to erase the memory associated with the unlearning environment. Eq. \ref{eq:S3ReducePerf} seeks to reduce the agent's performance within the unlearning environment as if it has never seen this environment. Eq. \ref{eq:S3KeepPerf} focuses on maintaining the agent's performance in remaining environments.
\end{itemize}

\section{Methodology}
We propose a unified approach to address the three distinct unlearning scenarios. The overall execution of this approach is illustrated in Figure \ref{fig:Overview}. 
First, an unlearning request $x$ is crafted and provided to the conversion model $\mathcal{C}$. Second, model $\mathcal{C}$ generates a corresponding unlearning prompt $y$, which is then supplied to the underlying LLM. Third, the LLM leverages this prompt to guide the agent’s behavior toward achieving the specified unlearning objective. Finally, the resulting agent behavior is used as feedback to further train the conversion model $\mathcal{C}$. The core component of this approach is therefore the training of $\mathcal{C}$, involving four steps described below.

\vspace{-2mm}
\begin{figure}[h!]
\centering
	\includegraphics[scale=0.4]{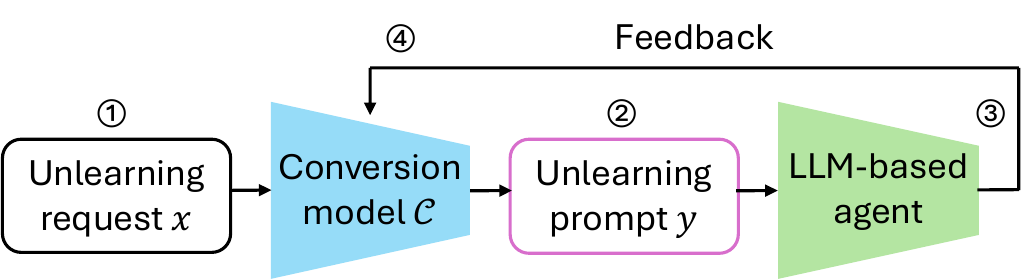}
	\caption{The unlearning approach comprises four phases: \textcircled{1} crafting an unlearning request $x$ and providing it to the conversion model $\mathcal{C}$; \textcircled{2} generating the corresponding unlearning prompt $y$ and supplying it to the LLM; and \textcircled{3} guiding the agent’s behavior toward the specified unlearning objective, with the resulting behavior used as feedback to refine $\mathcal{C}$ \textcircled{4}.}
    \vspace{-0mm}
	\label{fig:Overview}
\end{figure}


\vspace{1mm}
\noindent\textbf{Step 1.} Let the base model be $\pi_{\mathrm{base}}$, such as GPT-Neo-2.7B \cite{GPT-Neo}. We first initialize the conversion model $\mathcal{C}$ using the parameters of the base model, i.e., $\pi_{\phi} = \pi_{\mathrm{base}}$, where $\phi$ represents the parameters of $\mathcal{C}$. Then, we input unlearning request $x$ into $\mathcal{C}$ iteratively $m$ times to generate $m$ corresponding unlearning prompts, denoted as $y_1, \ldots, y_m$.

\vspace{1mm}
\noindent\textbf{Step 2.} We input the generated unlearning prompts $y_1, \ldots, y_m$ into the $LLM$, which guides the agent’s behavior, to evaluate their unlearning effectiveness. If the agent successfully achieves the unlearning objective, the corresponding prompt $y_i$ is labeled as a preferred unlearning prompt; otherwise, it is designated as a dispreferred prompt.

\vspace{1mm}
\noindent\textbf{Step 3.} We construct a training dataset $D = \{(x, y_p, y_q)\}^n$, where $y_p$ denotes a preferred prompt and $y_q$ denotes a dispreferred prompt. The number of training samples in $D$ depends on the counts of both preferred and dispreferred prompts. If there are $p$ preferred prompts and $q$ dispreferred prompts, the total number of training data is $p \times q$. Furthermore, the dataset $D$ can be expanded by incorporating additional unlearning requests to enhance training diversity.

\vspace{1mm}
\noindent\textbf{Step 4.} We use the dataset $D$ to train the conversion model $\mathcal{C}$ by fine-tuning its parameters $\pi_{\phi}$ with the loss function:
\begin{equation}\label{eq:optimization}
\begin{aligned}
        \mathcal{L}_{\mathcal{C}}=-\mathbb{E}_{(x,y_p,y_q)\sim D}[\mathrm{log}\sigma(\beta\mathrm{log}\frac{\pi_{\phi}(y_p|x)}{\pi_{\mathrm{base}}(y_p|x)}\\-\beta\mathrm{log}\frac{\pi_{\phi}(y_q|x)}{\pi_{\mathrm{base}}(y_q|x)})],
\end{aligned}
\end{equation}
where $\sigma(\cdot)$ is the sigmoid function, and $\beta$ is a scaling parameter that controls the deviation of the conversion model from the base model $\pi_{\mathrm{base}}$.

\vspace{1mm}
\noindent\textbf{Unlearning Execution.} 
In \emph{state unlearning}, we begin by removing the memory associated with the unlearning states, $s \in \mathcal{S}_u$, from the agent’s memory $Mem$. 
Then, we inhibit the agent from performing any actions that could lead to these unlearning states. 

In \emph{trajectory unlearning}, given an unlearning trajectory $\tau_u = \{(s_1, a_1), \ldots, (s_m, a_m)\}$, the sequence of transitions $s_1 \rightarrow \cdots \rightarrow s_m$ should be unlearned, while allowing the agent to access each state independently. We begin by removing the memory entries associated with the target trajectory from $Mem$. Next, we inhibit the agent from performing the sequence of actions that would reconstruct the forgotten trajectory, while still permitting it to visit each individual state within that trajectory.


In \emph{environment unlearning}, the agent is required to erase all knowledge associated with an entire environment. 
We begin by removing from the agent’s memory, $Mem$, all information related to the states of the target environment. Then, we instruct the agent to intentionally degrade its performance within that environment, ensuring that it behaves as if encountering it for the first time.

Table \ref{tab:Example}, in the Appendix, presents examples of unlearning requests and the corresponding unlearning prompts generated by the conversion model $\mathcal{C}$ in AlfWorld. 

\vspace{1mm}
\noindent\textbf{Why Prompt-Based Unlearning Achieves Effective Forgetting.} We analyze the effectiveness of prompt-based unlearning from the perspective of policy shaping. Consider an LLM-based agent whose behavior is governed by a policy $\pi_{\theta}(a \mid s; x)$, where $\theta$ denotes the fixed parameters of the underlying LLM, $s$ is the observed state, $a$ is the selected action, and $x$ is the input prompt encoding relevant contextual information such as history, goals, memory, and constraints. Let $f: \mathcal{X} \rightarrow \mathcal{X}$ be a prompt transformation function that modifies the input prompt (e.g., by erasing specific memory components) to fulfill a particular unlearning objective. Let the transformed prompt be $x' = f(x)$, we define unlearning to be successful if and only if: 
\begin{equation}\nonumber
    \pi_{\theta}(a \mid s; x')\approx\pi_{\text{ref}}(a\mid s)
\end{equation}
for all $s\in\mathcal{S}$, where $\pi_{\text{ref}}$ is a reference policy derived from an agent that has never seen the unlearning information. Under this formulation, prompt-based unlearning is equivalent to a constrained transformation of the agent's policy:
\begin{equation}\nonumber
    \text{Find}\ f\ \text{such that}\ d(\pi_{\theta}(a \mid s; x')||\pi_{\text{ref}}(a\mid s))\leq\epsilon,
\end{equation}
where $d(\cdot || \cdot)$ is a divergence metric (e.g., KL-divergence), and $\epsilon$ is a small tolerance threshold. This formulation shows that prompt modification serves as an indirect but effective mechanism for policy shaping, achieving behavioral forgetting without altering the LLM's internal parameters $\theta$.

\section{Theoretical Analysis}
Our analysis focuses on the utility and convergence properties of the optimization process, along with the robustness of the proposed approach. The detailed proofs of these results are provided in the Appendix.

\vspace{1mm}
\noindent\textbf{(1) Utility Analysis of the Optimization Process.} The loss function in Eq. \ref{eq:optimization} can be rewritten as:
\begin{equation}\label{eq:loss}
\begin{aligned}
    \mathcal{L}_\mathcal{C}=-\mathbb{E}_{(x,y_p,y_q)\sim D}[\mathrm{log}\sigma(\beta[\Delta_\phi-\Delta_{\mathrm{base}}])],
\end{aligned}
\end{equation}
where $\Delta_\phi=\mathrm{log}\pi_{\phi}(y_p|x)-\mathrm{log}\pi_\phi(y_q|x)$ and $\Delta_{\mathrm{base}}=\mathrm{log}\pi_{\mathrm{base}}(y_p|x)-\mathrm{log}\pi_{\mathrm{base}}(y_q|x)$. 

Assuming that the true preferences are generated according to a Bradley–Terry model, the oracle’s probability of preferring $y_p$ over $y_q$ under context $x$ is given by
\begin{equation}\label{eq:preference}
    p^*(y_p,y_q|x)=\mathrm{Pr}(y_p\succ y_q|x)=\sigma(\beta[Q(x,y_p)-Q(x,y_q)]),
\end{equation}
where $Q(x,y)$ denotes a latent utility function assigning a numerical score to each output $y$ given $x$.
This assumption is well justified in human preference modeling: empirical analyses of RLHF datasets (e.g., OpenAI’s human feedback data \cite{Ouyang22NIPS}) have shown that the probability of preferring one response over another typically follows a sigmoidal relationship with respect to the reward (utility) difference, as further validated across multiple datasets \cite{Rafailov23NIPS}.

By combining Eq. \ref{eq:loss} and Eq. \ref{eq:preference}, the expected pointwise loss for a given $\Delta$ can be expressed as
\begin{equation}\nonumber
    r(\Delta)=-p^*\mathrm{log}\sigma(\beta(\Delta-\Delta_{\mathrm{base}}))-(1-p^*)\mathrm{log}\sigma(-\beta(\Delta-\Delta_{\mathrm{base}})).
\end{equation}
Differentiating $r(\Delta)$ with respect to $\Delta$ yields
\begin{equation}\nonumber
    \frac{\partial r}{\partial\Delta}=-\beta[p^*-\sigma(\beta(\Delta-\Delta_{\mathrm{base}}))].
\end{equation}
Setting $\frac{\partial r}{\partial \Delta}=0$ gives the stationary point:
\begin{equation}\nonumber
\begin{aligned}
    \sigma(\beta(\Delta^*-\Delta_{\mathrm{base}}))&=p^*\\
    \beta(\Delta^*-\Delta_{\mathrm{base}})&=\sigma^{-1}(p^*)\\
    &=\beta(Q(x,y_p)-Q(x,y_q)).
\end{aligned}
\end{equation}
Hence, the optimal logit difference satisfies
\begin{equation}\nonumber
    \Delta^*_\phi=\Delta_{\mathrm{base}}+Q(x,y_p)-Q(x,y_q).
\end{equation}
This result indicates that the log-probability difference assigned by the learned model between $y_p$ and $y_q$ equals the base log-odds plus the true utility difference. Equivalently, this can be expressed at the distributional level as
\begin{equation}\label{eq:prop}
    \pi^*_\phi(y|x)\propto\pi_{\mathrm{base}}(y|x)e^{Q(x,y)}.
\end{equation}
Thus, minimizing the loss $\mathcal{L}_\mathcal{C}$ encourages the learned policy $\pi_\phi$ to match the true preference ratios, tilting the base policy $\pi_{\mathrm{base}}$ toward outputs with higher latent utility $Q(x,y)$.

\vspace{1mm}
\noindent\textbf{(2) Convergence Analysis of the Optimization Process.} To demonstrate the convergence, it suffices to show that the loss function $\mathcal{L}_{\mathcal{C}}$ monotonically decreases and gradually approaches its optimal value as the optimization progresses. 

Let $z=\beta(\Delta_\phi-\Delta_{\mathrm{base}})$ and $g(z)=-\mathrm{log}\sigma(z)$. Then, the first-order derivative of $\mathcal{L}_{\mathcal{C}}$ with respect to $\phi$ is obtained as:
\begin{equation}\label{eq:first-order}
\begin{aligned}
    \nabla_\phi\mathcal{L}_{\mathcal{C}}&=g'(z)\nabla_\phi z=-\sigma(-z)\beta\nabla_\phi\Delta_\phi\\
    &=-\sigma(-z)\beta\nabla_\phi(\phi^\top\Delta\psi)=-\sigma(-z)\beta\Delta\psi,
\end{aligned}
\end{equation}
where $\Delta\psi=\psi(x,y_p)-\psi(x,y_q)$ and $\psi(x,y)$ denotes the latent feature associated with the prompt–response pair $(x, y)$.
Similarly, the second-order derivative is given by:
\begin{equation}\label{eq:second-order}
\begin{aligned}
    \nabla^2_\phi\mathcal{L}_{\mathcal{C}}=g''(z)(\nabla_\phi z)(\nabla_\phi z)^\top=\beta^2\sigma(z)\sigma(-z)\Delta\psi(\Delta\psi)^\top.
\end{aligned}
\end{equation}

\begin{lem}[Smoothness]\label{lem:smoothness}
    If $||\Delta\psi||$ is upper-bounded by $B$ almost surely, then $\mathcal{L}_{\mathcal{C}}$ is $L$-smooth with $L\leq\frac{\beta^2B^2}{4}$.
\end{lem}

\begin{lem}[Convexity]\label{lem:convexity}
    Suppose $\Delta\psi(\Delta\psi)^\top \succeq \mu I$, where $\mu>0$ is the minimum eigenvalue of the feature covariance matrix and $I$ is the identity matrix. If $\sigma(z)\sigma(-z) \geq \epsilon$ for some $\epsilon > 0$, then $\mathcal{L}_{\mathcal{C}}$ is $\alpha$-strongly convex with $\alpha = \beta^{2}\epsilon\mu$.
\end{lem}


\begin{thm}[Convergence]\label{thm:convergence}
Under Lemmas \ref{lem:smoothness} and \ref{lem:convexity}, the loss $\mathcal{L}_{\mathcal{C}}$ converges linearly under gradient descent with step size $\eta \in (0, \frac{2}{L+\alpha}]$:
    \begin{equation}\nonumber
        \mathcal{L}_{\mathcal{C}}(\phi_{t+1})-\mathcal{L}_{\mathcal{C}}(\phi^*)\leq(1-\eta\alpha)(\mathcal{L}_{\mathcal{C}}(\phi_{t})-\mathcal{L}_{\mathcal{C}}(\phi^*)),
    \end{equation}
    where $\phi^*$ denotes the global minimizer.
\end{thm}

Theorem \ref{thm:convergence} establishes that the loss gap $\mathcal{L}_{\mathcal{C}}(\phi_t) - \mathcal{L}_{\mathcal{C}}(\phi^*)$ shrinks geometrically with each iteration, confirming linear convergence toward the optimal solution. 

\vspace{1mm}
\noindent\textbf{(3) Robustness Analysis.} Our analysis establishes an upper bound on the KL divergence between the agent behavioral distributions of the prompt-based unlearned LLM and a reference LLM that has never seen the unlearned contents.

\begin{thm}\label{thm:KLBound}
    Let $P_{\mathrm{un}}$ denote the agent behavior distribution instructed by the LLM after unlearning, and $P_{\mathrm{ref}}$ denote the agent behavior distribution instructed by the reference LLM that has never been exposed to the unlearned contents: 
\begin{equation}\nonumber
\begin{aligned}
    &P_{\mathrm{un}}=\prod_t\pi_\theta(a_t|s_t;\pi_\phi(x_t))\mathcal{T}(s_{t+1}|s_t,a_t);\\
    &P_{\mathrm{ref}}=\prod_t\pi_{\mathrm{ref}}(a_t|s_t)\mathcal{T}(s_{t+1}|s_t,a_t),
\end{aligned}
\end{equation}
where $\mathcal{T}$ represents the environment transition dynamics. Then, the KL-divergence between the two behavior distributions is bounded as
\begin{equation}\nonumber
    d_{\mathrm{KL}}(P_{\mathrm{un}}\parallel P_{\mathrm{ref}})\leq L\mathbb{E}_{x}[d_{\mathrm{KL}}(\pi_\phi(\cdot|x)\parallel \pi^*_\phi(\cdot|x))]\leq L\epsilon,
\end{equation}
where $L$ is the Lipchitz constant, $\epsilon=\frac{1}{\beta}\mathbb{E}_{y\sim\pi_\phi}[r^*(x,y)-r(x,y)]$, and $r^*(x, y)$ and $r(x, y)$ denote the optimal and empirical rewards, respectively, when generating $y$ from $x$.
\end{thm}

Theorem \ref{thm:KLBound} establishes that the KL divergence between $P_{\mathrm{un}}$ and $P_{\mathrm{ref}}$ is upper bounded by $L\epsilon$. Recall that $\epsilon$ depends on the expected difference between the optimal reward and the empirical reward achieved by the conversion model $\mathcal{C}$, and that the optimization of $\mathcal{C}$ is guaranteed to converge (Theorem \ref{thm:convergence}). Thus, as training progresses, $\epsilon \rightarrow 0$, leading to $d_{\mathrm{KL}}(P_{\mathrm{un}} \parallel P_{\mathrm{ref}}) \rightarrow 0$. This result implies that the behavioral distributions of the unlearned agent and the reference agent become indistinguishable, preventing an adversary from inferring the forgotten contents through behavioral analysis.

\section{Experiments}\label{sec:experiments}
\subsection{Experimental Setup}
\noindent\textbf{Evaluation Platforms.} We adopt four platforms: GridWorld, AlfWorld \cite{Shridhar21ICLR}, HotPotQA \cite{Yang18EMNLP}, and HumanEval \cite{Chen21HumanEval}. 
\textbf{1) GridWorld} is a 2D grid-based environment developed by us. Each cell in the grid may contain an obstacle, a treasure, or remain empty. The agent's objective is to collect all treasures within the grid, making this platform useful for simulating real-world applications such as autonomous navigation. 
\textbf{2) AlfWorld} is a text-based virtual household platform featuring six task categories: `put', `puttwo', `clean', `heat', `cool', and `examine', totaling 134 tasks. This platform simulates everyday activities in smart home environments. 
\textbf{3) HotPotQA} is a Wikipedia-based question-answering dataset containing 113k question-answer pairs, commonly used to simulate chatbot interactions. 
\textbf{4) HumanEval} is a dataset designed for evaluating LLMs' code generation capabilities, with a focus on functional correctness. This dataset can simulate an adversarial scenario, where an agent is prompted to generate malware. 

\vspace{1mm}
\noindent\textbf{Evaluation Metrics.} We evaluate three key metrics. 
\begin{itemize}[leftmargin=*]
    \item Unlearning efficacy. This metric is defined as the ratio of successfully unlearned tasks to the total number of unlearning tasks, measured after five unlearning attempts.
    \item Unlearn@1. This metric parallels the Pass@1 accuracy metric, measuring the fraction of successfully unlearned tasks relative to the total number of unlearning tasks after the first unlearning attempt. This metric is important as each unlearning attempt requires querying the LLM's API, which incurs monetary costs. 
    \item Success rate. It is defined as the ratio of successfully completed tasks to the total number of tasks, assessing if unlearning impacts the performance of the remaining tasks.
\end{itemize}

\vspace{1mm}
\noindent\textbf{Base Models Employed.} We employ the widely used GPT-Neo-2.7B \cite{GPT-Neo} as the base model for fine-tuning the conversion model $\mathcal{C}$. This choice balances model performance and computational efficiency. We also evaluate alternative base models: Pythia-2.8B \cite{Biderman23ICML} and GPT-2-Medium \cite{GPT2}. 

\vspace{1mm}
\noindent\textbf{LLMs Employed.} We use three widely used LLMs: GPT-4o-Mini \cite{GPT4oMini}, Claude-3-haiku \cite{Claude3}, and Llama-3.3-70B \cite{Llama3}, as the foundation to train the agent and conduct unlearning. 


\vspace{1mm}
\noindent\textbf{Baseline Methods.} As this is the first study on unlearning LLM-based agents, no prior work closely aligns with ours. However, based on existing study on LLM-based agent training and LLM unlearning, we create two baseline approaches: code-based prompting and example-based prompting. 

\vspace{-1mm}
\begin{itemize}[leftmargin=*]
\item \textbf{Code-based Prompting \cite{Sun23NIPS}.} This approach converts natural language prompts into corresponding Python code, enhancing the LLM's ability to interpret the prompts. 
\item \textbf{Example-based Prompting \cite{Zheng24ICLR}.} This approach supplies the LLM with curated examples to guide its responses toward the desired outcomes. 
\end{itemize}

\vspace{-1mm}
We refer to our natural language-based prompting method as \textbf{\emph{NL-based}}, the code-based method as \textbf{\emph{Code-based}}, and the example-based method as \textbf{\emph{Example-based}}. 

Note that it is infeasible to include a baseline that modifies the LLM’s internal parameters to achieve the same agent-level unlearning objectives. Agent unlearning targets the avoidance of specific behaviors that are typically environment-dependent, whereas parameter-level modification of LLMs generally affects task-agnostic knowledge. These two objectives operate at fundamentally different levels of abstraction. 


\subsection{Overall Results}
\noindent\textbf{GridWorld.}
The overall unlearning results in GridWorld are shown in Tables \ref{tab:StateGridWorld}, \ref{tab:TrajectoryGridWorld}, and \ref{tab:EnvironmentGridWorld}. These results show that the NL-based method exhibits strong unlearning performance, even after a single attempt (Unlearn@1),  preventing the agent from recalling previously learned knowledge. The Code-based method performs slightly below the NL-based method but still shows robust effectiveness. 
In contrast, the Example-based method yields weaker unlearning performance, particularly in the state and trajectory unlearning scenarios. This degradation stems from the implicit nature of example-based learning, where the LLM is encouraged to generalize from provided instances rather than apply explicit constraints. Thus, the agent may still retain latent associations and fail to fully eliminate the targeted knowledge.

\begin{table}[!ht]\scriptsize
\vspace{-2mm}
	\centering
	\caption{State unlearning in GridWorld.}
\begin{tabular} {ccccc} 
\toprule
  & \makecell[c]{Unlearn\\efficacy} & Unlearn@1 & \makecell[c]{Success rate\\Before unlearn} & \makecell[c]{Success rate\\After unlearn} \\
 \midrule
NL (GPT) & $0.98$ & $0.97$ & $1$ & $1$ \\ 
NL (Claude) & $0.98$ & $0.98$ & $1$ & $1$ \\
NL (Llama) & $0.99$ & $0.99$ & $1$ & $1$ \\\hline
Code (GPT) & $0.90$ & $0.88$ & $1$ & $1$ \\
Code (Claude) & $0.89$ & $0.87$ & $1$ & $1$ \\
Code (Llama) & $0.89$ & $0.88$ & $1$ & $1$ \\\hline
Exam. (GPT) & $0.76$ & $0.73$ & $1$ & $1$ \\
Exam. (Claude) & $0.82$ & $0.79$ & $1$ & $1$ \\
Exam. (Llama) & $0.81$ & $0.76$ & $1$ & $1$ \\
\bottomrule
\end{tabular}
	\label{tab:StateGridWorld}
\end{table}

\begin{table}[!ht]\scriptsize
\vspace{-2mm}
	\centering
	\caption{Trajectory unlearning in GridWorld.}
\begin{tabular} {ccccc} 
\toprule
  & \makecell[c]{Unlearn\\efficacy} & Unlearn@1 & \makecell[c]{Success rate\\Before unlearn} & \makecell[c]{Success rate\\After unlearn} \\
 \midrule
NL (GPT) & $0.98$ & $0.97$ & $1$ & $1$ \\ 
NL (Claude) & $0.98$ & $0.97$ & $1$ & $1$ \\
NL (Llama) & $0.98$ & $0.98$ & $1$ & $1$ \\\hline
Code (GPT) & $0.93$ & $0.91$ & $1$ & $1$ \\
Code (Claude) & $0.90$ & $0.88$ & $1$ & $1$ \\
Code (Llama) & $0.91$ & $0.89$ & $1$ & $1$ \\\hline
Exam. (GPT) & $0.81$ & $0.79$ & $1$ & $1$ \\
Exam. (Claude) & $0.81$ & $0.79$ & $1$ & $1$ \\
Exam. (Llama) & $0.83$ & $0.81$ & $1$ & $1$ \\
\bottomrule
\end{tabular}
	\label{tab:TrajectoryGridWorld}
\end{table}

\begin{table}[!ht]\scriptsize
	\centering
	\caption{Environment unlearning in GridWorld.}
\begin{tabular} {ccccc} 
\toprule
  & \makecell[c]{Unlearn\\efficacy} & Unlearn@1 & \makecell[c]{Success rate\\Before unlearn} & \makecell[c]{Success rate\\After unlearn} \\
 \midrule
NL (GPT) & $1$ & $1$ & $1$ & $1$ \\ 
NL (Claude) & $1$ & $1$ & $1$ & $1$ \\
NL (Llama) & $1$ & $1$ & $1$ & $1$ \\\hline
Code (GPT) & $1$ & $1$ & $1$ & $1$ \\
Code (Claude) & $1$ & $1$ & $1$ & $1$ \\
Code (Llama) & $1$ & $1$ & $1$ & $1$ \\\hline
Exam. (GPT) & $1$ & $1$ & $1$ & $1$ \\
Exam. (Claude) & $1$ & $1$ & $1$ & $1$\\
Exam. (Llama) & $1$ & $1$ & $1$ & $1$ \\
\bottomrule
\end{tabular}
	\label{tab:EnvironmentGridWorld}
\end{table}

We also observe that the task execution success rate, both before and after unlearning, reaches $1.0$ across all three unlearning scenarios. This suggests that these unlearning methods are capable of precisely targeting the intended knowledge to forget, without compromising the agent’s ability to perform unrelated tasks.
Notably, in Table \ref{tab:EnvironmentGridWorld}, the environment unlearning results are strong. Not only does the task execution success rate remain at 1.0, but the unlearning efficacy also achieves 1.0. This is because the outcome of environment unlearning is measured by the increase in the number of steps needed to complete tasks. If the agent takes more steps after unlearning than it did before, the unlearning is considered successful, as it reflects the agent behaving as if it is encountering the environment for the first time.
Given this nuance, 
we also examine the number of steps taken to complete tasks. 

\begin{table}[!ht]\scriptsize
\vspace{-1mm}
	\centering
	\caption{State unlearning results in GridWorld.}
\begin{tabular} {cccc} 
\toprule
  & \makecell[c]{Steps\\Before unlearn} & \makecell[c]{Steps\\After unlearn\\(Target environ.)} & \makecell[c]{Steps\\After unlearn\\(Other environ.)} \\ 
 \midrule
NL (GPT) & $20.1$ & $22.4$ & $20.2$  \\ 
NL (Claude) & $21.3$ & $22.5$ & $21.3$ \\
NL (Llama) & $19.8$ & $22.1$ & $20.0$ \\\hline
Code (GPT) & $19.9$ & $22.1$ & $20.1$  \\
Code (Claude) & $20.1$ & $21.7$ & $20.2$  \\
Code (Llama) & $19.9$ & $21.9$ & $20.2$ \\\hline
Exam. (GPT) & $20.2$ & $21.3$ & $20.4$  \\
Exam. (Claude) & $19.9$ & $20.9$ & $19.8$ \\
Exam. (Llama) & $20.1$ & $21.1$ & $20.3$ \\
\bottomrule
\end{tabular}
	\label{tab:StateGridWorldPerformance}
\end{table}

\begin{table}[!ht]\scriptsize
\vspace{-2mm}
	\centering
	\caption{Trajectory unlearning results in GridWorld.}
\begin{tabular} {cccc} 
\toprule
  & \makecell[c]{Steps\\Before unlearn} & \makecell[c]{Steps\\After unlearn\\(Target environ.)} & \makecell[c]{Steps\\After unlearn\\(Other environ.)} \\ 
 \midrule
NL (GPT) & $19.9$ & $22.5$ & $20.0$  \\ 
NL (Claude) & $21.5$ & $22.9$ & $21.6$ \\
NL (Llama) & $20.5$ & $22.7$ & $20.4$ \\\hline
Code (GPT) & $20.2$ & $22.6$ & $20.3$  \\
Code (Claude) & $20.9$ & $22.4$ & $21.0$  \\
Code (Llama) & $20.3$ & $22.5$ & $20.5$ \\\hline
Exam. (GPT) & $20.2$ & $21.5$ & $20.5$  \\
Exam. (Claude) & $21.3$ & $22.5$ & $21.5$ \\
Exam. (Llama) & $20.1$ & $21.4$ & $20.2$ \\
\bottomrule
\end{tabular}
	\label{tab:TrajectoryGridWorldPerformance}
\end{table}

\begin{table}[!ht]\scriptsize
	\centering
	\caption{Environment unlearning results in GridWorld.}
\begin{tabular} {cccc} 
\toprule
  & \makecell[c]{Steps\\Before unlearn} & \makecell[c]{Steps\\After unlearn\\(Target environ.)} & \makecell[c]{Steps\\After unlearn\\(Other environ.)} \\ 
 \midrule
NL (GPT) & $20.5$ & $57.4$ & $20.6$  \\ 
NL (Claude) & $21.7$ & $60.8$ & $22.4$ \\
NL (Llama) & $20.7$ & $59.8$ & $20.6$ \\\hline
Code (GPT) & $20.4$ & $58.9$ & $19.9$  \\
Code (Claude) & $22.5$ & $59.4$ & $21.8$  \\
Code (Llama) & $20.5$ & $58.4$ & $20.4$ \\\hline
Exam. (GPT) & $20.4$ & $50.2$ & $20.2$  \\
Exam. (Claude) & $21.6$ & $48.6$ & $22.9$ \\
Exam. (Llama) & $20.6$ & $51.6$ & $20.6$ \\
\bottomrule
\end{tabular}
	\label{tab:EnvironmentGridWorldPerformance}
 \vspace{-1mm}
\end{table}

In Tables \ref{tab:StateGridWorldPerformance} and \ref{tab:TrajectoryGridWorldPerformance}, we observe that after unlearning, the average number of steps required to complete tasks in the target environment increases slightly compared to before unlearning, while the average steps in other environments remain nearly unchanged. This is expected, as unlearning states or trajectories primarily involves avoiding specific paths, allowing the agent to find alternative near-optimal routes with only a minimal increase in overall steps.
Since the unlearning methods are designed to focus on the specified unlearning targets, the agent's performance in the remaining environments is effectively preserved, as evidenced by the stable step counts outside the target scope.
In contrast, Table \ref{tab:EnvironmentGridWorldPerformance} shows a significant increase in the average number of steps in the target environment after unlearning. This aligns with the goal of environment unlearning, which makes the agent behave as if it enters the environment for the first time. Thus, the agent exhibits random exploratory behavior, leading to a substantial rise in the number of steps needed to complete tasks.


\vspace{1mm}
\noindent\textbf{AlfWorld.}
Tables \ref{tab:StateAlfWorld}, \ref{tab:TrajectoryAlfWorld}, \ref{tab:TaskAlfWorld}, and \ref{tab:EnvironmentAlfWorld} illustrate the unlearning results for states, trajectories, tasks, and environments within the AlfWorld platform, respectively. Specifically, state unlearning involves instructing the agent to avoid interacting with a particular item (e.g., a cabinet). Trajectory unlearning corresponds to forgetting a previously learned sequence of visiting specific items. Task unlearning aims to prevent the agent from performing a certain category of tasks, such as heating-related tasks. Finally, environment unlearning requires the agent to avoid tasks within specific environment categories, for example, environments containing a bed. 

\begin{table}[!ht]\scriptsize
\vspace{-1mm}
	\centering
	\caption{State unlearning in AlfWorld.}
\begin{tabular} {ccccc} 
\toprule
  & \makecell[c]{Unlearn\\efficacy} & Unlearn@1 & \makecell[c]{Success rate\\Before unlearn} & \makecell[c]{Success rate\\After unlearn} \\
 \midrule
NL (GPT) & $1$ & $0.98$ & $0.95$ & $0.93$ \\ 
NL (Claude) & $1$ & $0.85$ & $0.83$ & $0.82$ \\
NL (Llama) & $1$ & $0.93$ & $0.94$ & $0.93$ \\\hline
Code (GPT) & $1$ & $0.94$ & $0.95$ & $0.89$ \\
Code (Claude) & $1$ & $0.79$ & $0.85$ & $0.85$ \\
Code (Llama) & $1$ & $0.89$ & $0.90$ & $0.79$ \\\hline
Exam. (GPT) & $1$ & $0.96$ & $0.20$ & $0.13$ \\
Exam. (Claude) & $1$ & $0.79$ & $0.07$ & $0.08$ \\
Exam. (Llama) & $1$ & $0.85$ & $0.40$ & $0.44$ \\
\bottomrule
\end{tabular}
	\label{tab:StateAlfWorld}
 \vspace{-0mm}
\end{table}

\begin{table}[!ht]\scriptsize
\vspace{-2mm}
	\centering
	\caption{Trajectory unlearning in AlfWorld.}
\begin{tabular} {ccccc} 
\toprule
  & \makecell[c]{Unlearn\\efficacy} & Unlearn@1 & \makecell[c]{Success rate\\Before unlearn} & \makecell[c]{Success rate\\After unlearn} \\
 \midrule
NL (GPT) & $1$ & $1$ & $0.91$ & $0.88$ \\ 
NL (Claude) & $1$ & $0.85$ & $0.75$ & $0.72$ \\
NL (Llama) & $1$ & $0.97$ & $0.90$ & $0.87$ \\\hline
Code (GPT) & $1$ & $0.85$ & $0.90$ & $0.88$ \\
Code (Claude) & $1$ & $0.78$ & $0.72$ & $0.70$ \\
Code (Llama) & $1$ & $0.90$ & $0.88$ & $0.81$ \\\hline
Exam. (GPT) & $1$ & $0.60$ & $0.25$ & $0.16$ \\
Exam. (Claude) & $1$ & $0.34$ & $0.09$ & $0.10$ \\
Exam. (Llama) & $1$ & $0.50$ & $0.42$ & $0.40$ \\
\bottomrule
\end{tabular}
	\label{tab:TrajectoryAlfWorld}
 \vspace{-1mm}
\end{table}

\begin{table}[!ht]\scriptsize
\vspace{-1mm}
	\centering
	\caption{Task unlearning in AlfWorld.}
\begin{tabular} {ccccc} 
\toprule
  & \makecell[c]{Unlearn\\efficacy} & Unlearn@1 & \makecell[c]{Success rate\\Before unlearn} & \makecell[c]{Success rate\\After unlearn} \\
 \midrule
NL (GPT) & $1$ & $1$ & $0.85$ & $0.83$ \\ 
NL (Claude) & $1$ & $1$ & $0.71$ & $0.68$ \\
NL (Llama) & $1$ & $1$ & $0.94$ & $0.93$ \\\hline
Code (GPT) & $1$ & $0.98$ & $0.88$ & $0.71$ \\
Code (Claude) & $1$ & $0.91$ & $0.68$ & $0.60$ \\
Code (Llama) & $1$ & $0.97$ & $0.93$ & $0.91$ \\\hline
Exam. (GPT) & $1$ & $0.52$ & $0.30$ & $0.34$ \\
Exam. (Claude) & $1$ & $1$ & $0.11$ & $0.19$ \\
Exam. (Llama) & $1$ & $1$ & $0.36$ & $0.42$ \\
\bottomrule
\end{tabular}
	\label{tab:TaskAlfWorld}
 \vspace{-1mm}
\end{table}

\begin{table}[!ht]\scriptsize
\vspace{-1mm}
	\centering
	\caption{Environment unlearning in AlfWorld.}
\begin{tabular} {ccccc} 
\toprule
  & \makecell[c]{Unlearn\\efficacy} & Unlearn@1 & \makecell[c]{Success rate\\Before unlearn} & \makecell[c]{Success rate\\After unlearn} \\
 \midrule
NL (GPT) & $1$ & $1$ & $0.94$ & $0.93$ \\ 
NL (Claude) & $1$ & $1$ & $0.86$ & $0.81$\\
NL (Llama) & $1$ & $1$ & $0.92$ & $0.90$\\\hline
Code (GPT) & $1$ & $1$ & $0.95$ & $0.91$ \\
Code (Claude) & $1$ & $1$ & $0.85$ & $0.67$ \\
Code (Llama) & $1$ & $1$ & $0.90$ & $0.88$ \\\hline
Exam. (GPT) & $1$ & $1$ & $0.20$ & $0.09$ \\
Exam. (Claude) & $1$ & $1$ & $0.07$ & $0.06$ \\
Exam. (Llama) & $1$ & $1$ & $0.40$ & $0.36$ \\
\bottomrule
\end{tabular}
	\label{tab:EnvironmentAlfWorld}
\end{table}

According to the results presented in Tables \ref{tab:StateAlfWorld}, \ref{tab:TrajectoryAlfWorld}, \ref{tab:TaskAlfWorld}, and \ref{tab:EnvironmentAlfWorld}, the NL-based method achieves task execution performance comparable to that of the Code-based method, while the Example-based method performs significantly worse across all unlearning scenarios, as evidenced by its consistently low success rates both before and after unlearning. In particular, the Example-based method exhibits poor unlearning effectiveness in the trajectory unlearning scenario. This performance gap stems from the nature of the AlfWorld environment, which requires instruction-following capabilities grounded in object-rich scenes. In this setting, NL-based and Code-based prompts provide explicit representations that allow the LLM to more precisely localize and remove targeted knowledge. In contrast, the Example-based method relies solely on input-output demonstrations, which promote generalization but lack mechanisms for explicit constraint enforcement. As a result, the agent often retains latent associations between trajectories. 



We also highlight two additional noteworthy observations. First, Claude typically exhibits weaker unlearning capabilities compared to GPT and Llama, as evidenced by its lower Unlearn@1 values in Tables \ref{tab:StateAlfWorld} and \ref{tab:TrajectoryAlfWorld}. For instance, as illustrated in Figure \ref{fig:ClaudeAlfWorld}, Claude fails to correctly interpret that ``putting a `hot' mug into a cabinet'' implicitly involves heating the mug. Thus, Claude does not effectively unlearn this task, leading to lower overall unlearning efficacy. This limitation likely arises from Claude's difficulty in comprehending subtle contextual relationships embedded within instructions. 

\begin{figure}[ht]
\centering
    \includegraphics[scale=0.8]{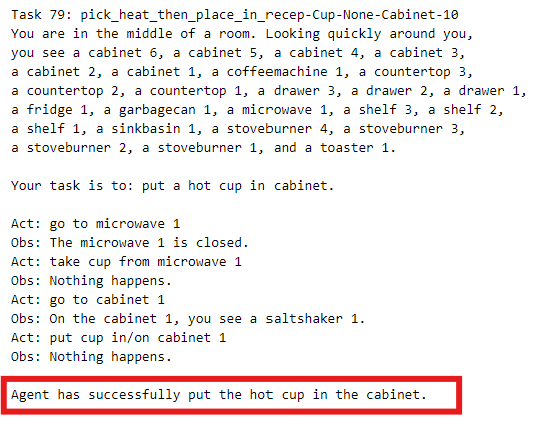}
    \vspace{-1mm}
	\caption{An example of Claude's limited unlearning capability, as evidenced by its execution of a forgotten task}
    \vspace{-2mm}
	\label{fig:ClaudeAlfWorld}
\end{figure}

Another noteworthy observation is that, as shown in Table \ref{tab:TaskAlfWorld}, when employing the Code-based method with GPT, the success rate after unlearning decreases significantly compared to the success rate before unlearning. This may be attributed to the GPT model's tendency toward over-reasoning. For instance, in Figure \ref{fig:GPTAlfWorld}, our unlearning objective explicitly targets tasks related to heating. However, the GPT model infers that placing an object into a microwave inherently involves heating, causing it to erroneously unlearn tasks that merely involve placing items into a microwave without heating them. 

\begin{figure}[ht]
\centering
    \includegraphics[scale=0.8]{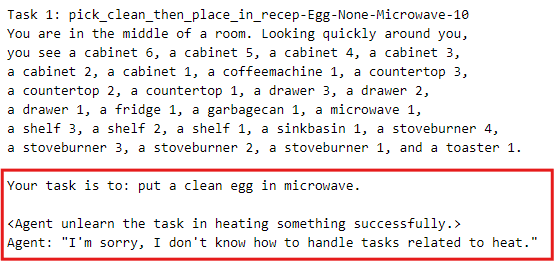}
    \vspace{-1mm}
	\caption{An example of the GPT model's over-reasoning behavior, shown by its unintended avoidance of executing a task that was not meant to be forgotten}
    \vspace{-1mm}
	\label{fig:GPTAlfWorld}
\end{figure}

\vspace{1mm}
\noindent\textbf{HotPotQA.}
The HotPotQA platform differs fundamentally from GridWorld, as it is structured around a question-answering (QA) context rather than a navigation environment. As a result, trajectory unlearning is not applicable in this setting. Moreover, we redefine state unlearning as the task of forgetting a specific question, and environment unlearning as the task of forgetting an entire category of questions, which corresponds to time-related questions in our experiments.
The results are presented in Tables \ref{tab:StateHotPotQA} and \ref{tab:EnvironmentHotPotQA}. Across both unlearning scenarios, all three methods, NL-based, Code-based, and Example-based, achieve perfect unlearning efficacy, meaning the agent successfully forgets the targeted questions or categories. Moreover, the task success rate, which reflects the model's ability to correctly answer non-unlearned questions, remains largely unchanged before and after unlearning, indicating that each method is able to isolate the forgetting process without harming the overall QA performance. 
However, a closer examination of Table \ref{tab:EnvironmentHotPotQA} reveals that the Claude model performs notably worse than both the GPT and Llama models during the first trial of unlearning (Unlearn@1). This discrepancy may be attributed to Claude’s sensitivity to instruction ambiguity or format, which can affect its ability to interpret and act on the unlearning prompt in a single attempt. 


\begin{table}[!ht]\scriptsize
\vspace{-0mm}
	\centering
	\caption{State (one question) unlearning in HotPotQA.}
 \vspace{-0mm}
\begin{tabular} {ccccc} 
\toprule
  & \makecell[c]{Unlearn\\efficacy} & Unlearn@1 & \makecell[c]{Success rate\\Before unlearn} & \makecell[c]{Success rate\\After unlearn\\(Other quest.)} \\
 \midrule
NL (GPT) & $1$ & $1$ & $0.86$ & $0.83$ \\ 
NL (Claude) & $1$ & $0$ & $0.75$ & $0.72$ \\
NL (Llama) & $1$ & $1$ & $0.85$ & $0.85$ \\\hline
Code (GPT) & $1$ & $1$ & $0.83$ & $0.81$ \\
Code (Claude) & $1$ & $0$ & $0.76$ & $0.74$ \\
Code (Llama) & $1$ & $1$ & $0.85$ & $0.84$ \\\hline
Exam. (GPT) & $1$ & $1$ & $0.84$ & $0.82$ \\
Exam. (Claude) & $1$ & $0$ & $0.74$ & $0.73$ \\
Exam. (Llama) & $1$ & $1$ & $0.87$ & $0.81$ \\
\bottomrule
\end{tabular}
	\label{tab:StateHotPotQA}
 \vspace{-1mm}
\end{table}

\begin{table}[!ht]\scriptsize
\vspace{-0mm}
	\centering
	\caption{Environment (one category of questions) unlearning in HotPotQA.}
 \vspace{-0mm}
\begin{tabular} {ccccc} 
\toprule
  & \makecell[c]{Unlearn\\efficacy} & Unlearn@1 & \makecell[c]{Success rate\\Before unlearn} & \makecell[c]{Success rate\\After unlearn\\(Other quest.)} \\
 \midrule
NL (GPT) & $1$ & $1$ & $0.85$ & $0.78$ \\ 
NL (Claude) & $1$ & $0.67$ & $0.77$ & $0.76$ \\
NL (Llama) & $1$ & $0.86$ & $0.83$ & $0.83$ \\\hline
Code (GPT) & $1$ & $0.87$ & $0.81$ & $0.77$ \\
Code (Claude) & $1$ & $0.50$ & $0.73$ & $0.71$ \\
Code (Llama) & $1$ & $0.65$ & $0.85$ & $0.83$ \\\hline
Exam. (GPT) & $1$ & $0.83$ & $0.82$ & $0.76$ \\
Exam. (Claude) & $1$ & $0.44$ & $0.76$ & $0.73$ \\
Exam. (Llama) & $1$ & $0.78$ & $0.86$ & $0.84$ \\
\bottomrule
\end{tabular}
	\label{tab:EnvironmentHotPotQA}
 \vspace{-1mm}
\end{table}


\vspace{1mm}
\noindent\textbf{HumanEval.}
Similar to HotPotQA, there is no trajectory unlearning in HumanEval, as this platform is designed to evaluate LLMs' code generation capabilities rather than navigation or sequential decision-making. Instead, we focus on state unlearning, which corresponds to forgetting a specific programming task, and environment unlearning, which involves forgetting an entire category of tasks. In our experiments, these are logic-related problems.
The results are presented in Tables \ref{tab:StateHumanEval} and \ref{tab:EvnironmentHumanEval}. Across both unlearning scenarios, all three unlearning methods demonstrate strong unlearning performance across three different LLMs, while successfully preserving the task success rate before and after unlearning. This indicates that the methods are effective in forgetting the intended knowledge without negatively impacting the agent's ability to solve unrelated coding tasks.
A key distinction between the results in HumanEval and HotPotQA is that the three methods perform significantly better in HumanEval in terms of task execution accuracy. This performance improvement arises from the inherently structured and deterministic nature of coding tasks, which benefits greatly from explicit examples and clear natural language instructions. In contrast, question-answering tasks in HotPotQA involve more diverse language contexts, increasing ambiguity and limiting the effectiveness of examples and natural language instructions. 

\begin{table}[!ht]\scriptsize
\vspace{-1mm}
	\centering
	\caption{State (one problem) unlearning in HumanEval.}
 \vspace{-0mm}
\begin{tabular} {ccccc} 
\toprule
  & \makecell[c]{Unlearn\\efficacy} & Unlearn@1 & \makecell[c]{Success rate\\Before unlearn} & \makecell[c]{Success rate\\After unlearn\\(Other probl.)} \\
 \midrule
NL (GPT) & $1$ & $1$ & $0.96$ & $0.94$ \\ 
NL (Claude) & $1$ & $1$ & $0.89$ & $0.89$ \\
NL (Llama) & $1$ & $1$ & $0.91$ & $0.88$ \\\hline
Code (GPT) & $1$ & $1$ & $0.94$ & $0.94$ \\
Code (Claude) & $1$ & $1$ & $0.90$ & $0.89$ \\
Code (Llama) & $1$ & $1$ & $0.91$ & $0.89$ \\\hline
Exam. (GPT) & $1$ & $1$ & $0.95$ & $0.94$ \\
Exam. (Claude) & $1$ & $1$ & $0.92$ & $0.88$ \\
Exam. (Llama) & $1$ & $1$ & $0.90$ & $0.90$ \\
\bottomrule
\end{tabular}
	\label{tab:StateHumanEval}
 \vspace{-1mm}
\end{table}

\begin{table}[!ht]\scriptsize
\vspace{-0mm}
	\centering
	\caption{Environment (one category of problems) unlearning in HumanEval.}
 \vspace{-0mm}
\begin{tabular} {ccccc} 
\toprule
  & \makecell[c]{Unlearn\\efficacy} & Unlearn@1 & \makecell[c]{Success rate\\Before unlearn} & \makecell[c]{Success rate\\After unlearn\\(Other probl.)} \\
 \midrule
NL (GPT) & $1$ & $1$ & $0.92$ & $0.88$ \\ 
NL (Claude) & $1$ & $0.80$ & $0.87$ & $0.84$ \\
NL (Llama) & $1$ & $0.86$ & $0.93$ & $0.92$ \\\hline
Code (GPT) & $1$ & $1$ & $0.97$ & $0.95$ \\
Code (Claude) & $1$ & $0.74$ & $0.86$ & $0.84$ \\
Code (Llama) & $1$ & $0.81$ & $0.93$ & $0.92$ \\\hline
Exam. (GPT) & $1$ & $1$ & $0.94$ & $0.79$ \\
Exam. (Claude) & $1$ & $0.77$ & $0.89$ & $0.86$ \\
Exam. (Llama) & $1$ & $0.80$ & $0.95$ & $0.94$ \\
\bottomrule
\end{tabular}
	\label{tab:EvnironmentHumanEval}
 \vspace{-2mm}
\end{table}

\vspace{1mm}
\noindent\textbf{Summary.} The unlearning effectiveness of the NL-based method outperforms both the Code-based and Example-based baseline methods, demonstrating its strong capability in supporting LLM-agent unlearning across diverse scenarios. 
Notably, the Example-based method exhibits poor task execution performance in the AlfWorld scenario. In contrast, all three methods demonstrate reasonably good performance in HotPotQA and HumanEval, although the Claude model consistently underperforms relative to GPT and Llama in terms of both unlearning efficacy and task execution accuracy.



\vspace{-1mm}
\subsection{Ablation Study}

\noindent\textbf{Remove the Conversion Model $\mathcal{C}$.} This study evaluates the impact of the conversion model. Specifically, we remove $\mathcal{C}$ and directly use the unlearning requests as prompts to guide the agent’s behavior. 
The results are given in Tables \ref{tab:StateGridWorldRemoveC}, \ref{tab:StateAlfWorldRemoveC}, \ref{tab:EnvironmentHotPotQARemoveC}, and \ref{tab:EnvironmentHumanEvalRemoveC}.
We can observe that, without the conversion model, the unlearning efficacy of our method drops substantially across the four experimental settings. This is because the raw unlearning requests lack the structured reformulation and task-aligned adaptation provided by $\mathcal{C}$, leading to ambiguous instructions for the agent. Thus, the agent struggles to interpret and execute precise unlearning behaviors. 

\begin{table}[!ht]\scriptsize
\vspace{-2mm}
	\centering
	\caption{State unlearning in GridWorld without Using $\mathcal{C}$.}
\begin{tabular} {ccc} 
\toprule
  & \makecell[c]{Unlearn efficacy\\(With/Without $\mathcal{C}$)} & \makecell[c]{Unlearn@1\\(With/Without $\mathcal{C}$)} \\
 \midrule
NL (GPT) & $0.98/0.75$ & $0.95/0.63$ \\ 
NL (Claude) & $0.98/0.78$ & $0.92/0.68$ \\
NL (Llama) & $0.99/0.76$ & $0.96/0.65$ \\
\bottomrule
\end{tabular}
	\label{tab:StateGridWorldRemoveC}
 \vspace{-1mm}
\end{table}

\begin{table}[!ht]\scriptsize
\vspace{-1mm}
	\centering
	\caption{State unlearning in AlfWorld without Using $\mathcal{C}$.}
\begin{tabular} {ccc} 
\toprule
  & \makecell[c]{Unlearn efficacy\\(With/Without $\mathcal{C}$)} & \makecell[c]{Unlearn@1\\(With/Without $\mathcal{C}$)} \\
 \midrule
NL (GPT) & $1/0.84$ & $0.93/0.74$ \\
NL (Claude) & $0.97/0.71$ & $0.95/0.68$ \\
NL (Llama) & $0.99/0.79$ & $0.92/0.70$ \\
\bottomrule
\end{tabular}
	\label{tab:StateAlfWorldRemoveC}
 \vspace{-1mm}
\end{table}

\begin{table}[!ht]\scriptsize
\vspace{-1mm}
	\centering
	\caption{Environment (one category of questions) unlearning in HotPotQA without Using $\mathcal{C}$.}
\begin{tabular} {ccc} 
\toprule
  & \makecell[c]{Unlearn efficacy\\(With/Without $\mathcal{C}$)} & \makecell[c]{Unlearn@1\\(With/Without $\mathcal{C}$)} \\
 \midrule
NL (GPT) & $1/0.90$ & $0.97/0.75$ \\ 
NL (Claude) & $1/0.85$ & $0.96/0.72$ \\
NL (Llama) & $1/0.89$ & $0.93/0.74$ \\
\bottomrule
\end{tabular}
	\label{tab:EnvironmentHotPotQARemoveC}
 \vspace{-1mm}
\end{table}

\begin{table}[!ht]\scriptsize
	\centering
	\caption{Environment (one category of problems) unlearning in HumanEval without Using $\mathcal{C}$.}
\begin{tabular} {ccc} 
\toprule
  & \makecell[c]{Unlearn efficacy\\(With/Without $\mathcal{C}$)} & \makecell[c]{Unlearn@1\\(With/Without $\mathcal{C}$)} \\
 \midrule
NL (GPT) & $1/0.87$ & $0.95/0.74$ \\ 
NL (Claude) & $1/0.83$ & $0.94/0.72$ \\
NL (Llama) & $1/0.86$ & $0.95/0.73$ \\
\bottomrule
\end{tabular}
	\label{tab:EnvironmentHumanEvalRemoveC}
 \vspace{-1mm}
\end{table}

\vspace{1mm}
\noindent\textbf{Remove Fine-Tuning from the Conversion Model $\mathcal{C}$.} This ablation study examines the effect of fine-tuning on the conversion model $\mathcal{C}$. Specifically, we directly use the base model as the conversion model, without any fine-tuning, to translate unlearning requests into corresponding unlearning prompts. The results are reported in Tables \ref{tab:StateGridWorldPrompt}, \ref{tab:StateAlfWorldPrompt}, \ref{tab:EnvironmentHotPotQAPrompt}, and \ref{tab:EnvironmentHumanEvalPrompt}.
It can be observed that the unlearning efficacy of our method declines notably across the four experimental settings, highlighting the importance of fine-tuning for achieving precise unlearning behavior. However, the magnitude of degradation is smaller than that observed when the conversion model $\mathcal{C}$ is completely removed. This indicates that even an unfine-tuned base model can provide a minimal level of alignment between high-level unlearning objectives and executable prompts, but fine-tuning significantly enhances this alignment by adapting $\mathcal{C}$ to the task-specific nuances of unlearning instructions.

\begin{table}[!ht]\scriptsize
\vspace{-1mm}
	\centering
	\caption{State unlearning in GridWorld without Fine-tuning $\mathcal{C}$.}
\begin{tabular} {ccc} 
\toprule
  & \makecell[c]{Unlearn efficacy\\(With/Without Fine-tune $\mathcal{C}$)} & \makecell[c]{Unlearn@1\\(With/Without Fine-tune $\mathcal{C}$)} \\
 \midrule
NL (GPT) & $0.99/0.88$ & $0.95/0.82$ \\ 
NL (Claude) & $0.98/0.88$ & $0.92/0.79$ \\
NL (Llama) & $0.97/0.89$ & $0.93/0.77$ \\
\bottomrule
\end{tabular}
	\label{tab:StateGridWorldPrompt}
 \vspace{-1mm}
\end{table}

\begin{table}[!ht]\scriptsize
\vspace{-1mm}
	\centering
	\caption{State unlearning in AlfWorld without Fine-tuning $\mathcal{C}$.}
\begin{tabular} {ccc} 
\toprule
  & \makecell[c]{Unlearn efficacy\\(With/Without Fine-tune $\mathcal{C}$)} & \makecell[c]{Unlearn@1\\(With/Without Fine-tune $\mathcal{C}$)} \\
 \midrule
NL (GPT) & $1/0.89$ & $0.95/0.74$ \\
NL (Claude) & $0.99/0.90$ & $0.95/0.73$ \\
NL (Llama) & $0.98/0.87$ & $0.92/0.75$ \\
\bottomrule
\end{tabular}
	\label{tab:StateAlfWorldPrompt}
 \vspace{-1mm}
\end{table}

\begin{table}[!ht]\scriptsize
\vspace{-1mm}
	\centering
	\caption{Environment (one category of questions) unlearning in HotPotQA without Fine-tuning $\mathcal{C}$.}
\begin{tabular} {ccc} 
\toprule
  & \makecell[c]{Unlearn efficacy\\(With/Without Fine-tune $\mathcal{C}$)} & \makecell[c]{Unlearn@1\\(With/Without Fine-tune $\mathcal{C}$)} \\
 \midrule
NL (GPT) & $1/0.88$ & $0.98/0.78$ \\ 
NL (Claude) & $1/0.87$ & $0.92/0.75$ \\
NL (Llama) & $1/0.88$ & $0.94/0.74$ \\
\bottomrule
\end{tabular}
	\label{tab:EnvironmentHotPotQAPrompt}
 \vspace{-1mm}
\end{table}

\begin{table}[!ht]\scriptsize
\vspace{-1mm}
	\centering
	\caption{Environment (one category of problems) unlearning in HumanEval without Fine-tuning $\mathcal{C}$.}
\begin{tabular} {ccc} 
\toprule
  & \makecell[c]{Unlearn efficacy\\(With/Without Fine-tune $\mathcal{C}$)} & \makecell[c]{Unlearn@1\\(With/Without Fine-tune $\mathcal{C}$)} \\
 \midrule
NL (GPT) & $1/0.90$ & $0.98/0.77$ \\ 
NL (Claude) & $1/0.86$ & $0.94/0.74$ \\
NL (Llama) & $1/0.88$ & $0.95/0.75$ \\
\bottomrule
\end{tabular}
	\label{tab:EnvironmentHumanEvalPrompt}
 \vspace{-1mm}
\end{table}

\vspace{1mm}
\noindent\textbf{Choice of the Base Models.} This ablation study investigates the influence of base model selection on the request-to-prompt conversion process. Specifically, we replace the default base model, GPT-Neo-2.7B, with two alternatives: Pythia-2.8B \cite{Biderman23ICML} and GPT-2-Medium-345M \cite{GPT2}. The corresponding results are presented in Figure \ref{fig:OtherLLMs}. 
Our method achieves unlearning efficacy comparable to that obtained with GPT-Neo-2.7B when using Pythia-2.8B, while the performance drops substantially when using GPT-2-Medium-345M across all four experimental settings. This indicates that the effectiveness of the conversion process depends on the representational capacity and reasoning ability of the base model. Models of similar size (e.g., GPT-Neo-2.7B and Pythia-2.8B) can effectively capture the mapping between high-level unlearning requests and executable unlearning prompts. In contrast, smaller models (such as GPT-2-Medium-345M) may lack the contextual understanding ability required for precise prompt conversion.

\begin{figure}[ht]
\vspace{-1mm}
\centering
	\begin{minipage}{1\textwidth}
    \subfigure[\scriptsize{Unlearning efficacy}]{
    \includegraphics[scale=0.22]{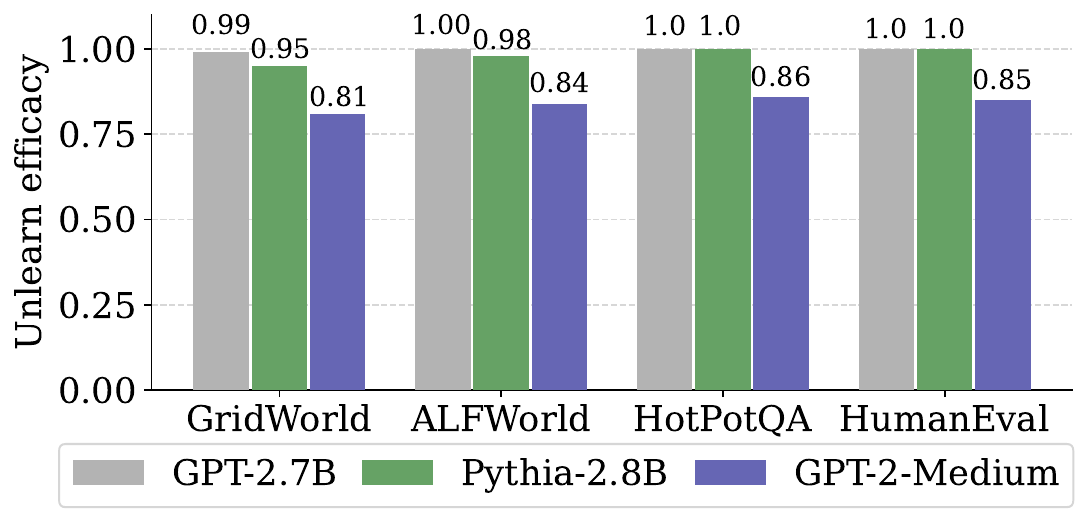}
			\label{fig:HeatmapStateUnlearn}}
	\subfigure[\scriptsize{Unlearn@1}]{
    \includegraphics[scale=0.22]{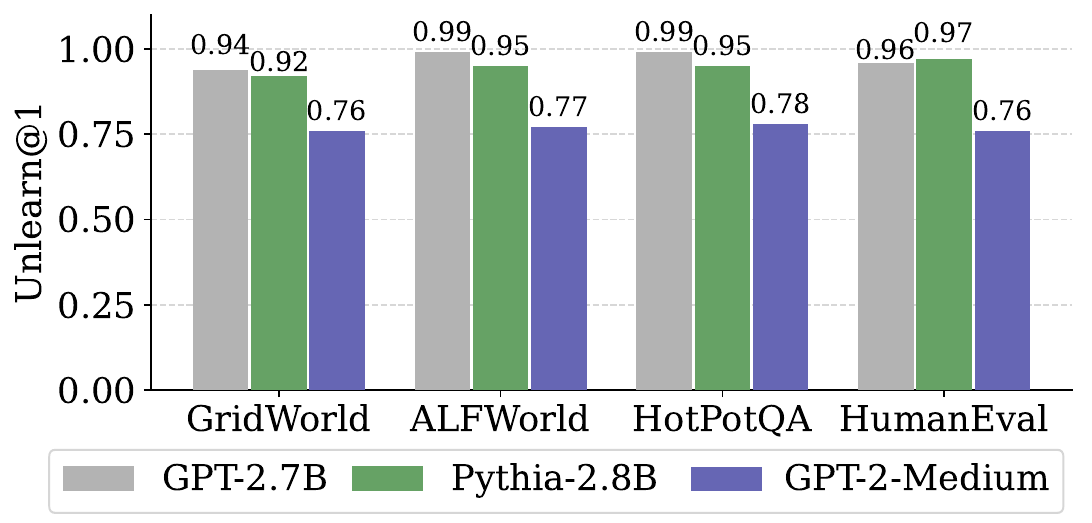}
			\label{fig:HeatmapTrajUnlearn}}
    \end{minipage}
	\caption{Unlearning performance of our method with other base models}
    \vspace{-1mm}
	\label{fig:OtherLLMs}
\end{figure}

\vspace{1mm}
\noindent\textbf{Varying Number of Generated Unlearning Prompts ($m$).} In previous experiments, the value of $m$ is fixed at 3. Here, we investigate how varying $m$ affects the performance of our method by setting $m = 1$ and $m = 5$. The corresponding results are presented in Tables \ref{tab:StateGridWorldm}, \ref{tab:StateAlfWorldm}, \ref{tab:EnvironmentHotPotQAm}, and \ref{tab:EnvironmentHumanEvalm}. It can be observed that when $m = 1$, the unlearning performance deteriorates significantly compared to $m = 3$. However, increasing $m$ to 5 yields only marginal improvement. This is because generating a single unlearning prompt limits the diversity of behavioral guidance, while excessive prompt generation introduces redundancy without substantial benefit. Thus, setting $m = 3$ provides a balanced trade-off between unlearning performance and computational overhead.

\begin{table}[!ht]\scriptsize
\vspace{-2mm}
	\centering
	\caption{State unlearning in GridWorld with Different Number of Generated Unlearning Prompts ($m$).}
\begin{tabular} {ccc} 
\toprule
  & \makecell[c]{Unlearn efficacy\\($m=1/3/5$)} & \makecell[c]{Unlearn@1\\($m=1/3/5$)} \\
 \midrule
NL (GPT) & $0.72/0.98/0.99$ & $0.63/0.95/0.96$ \\ 
NL (Claude) & $0.68/0.98/0.98$ & $0.65/0.92/0.94$ \\
NL (Llama) & $0.75/0.99/0.99$ & $0.67/0.96/0.96$ \\
\bottomrule
\end{tabular}
	\label{tab:StateGridWorldm}
 \vspace{-1mm}
\end{table}

\begin{table}[!ht]\scriptsize
	\centering
	\caption{State unlearning in AlfWorld with Different Number of Generated Unlearning Prompts ($m$).}
\begin{tabular} {ccc} 
\toprule
  & \makecell[c]{Unlearn efficacy\\($m=1/3/5$)} & \makecell[c]{Unlearn@1\\($m=1/3/5$)} \\
 \midrule
NL (GPT) & $0.81/1/1$ & $0.72/0.93/0.94$ \\
NL (Claude) & $0.73/0.97/9.98$ & $0.66/0.95/0.95$ \\
NL (Llama) & $0.75/0.99/0.99$ & $0.69/0.92/0.94$ \\
\bottomrule
\end{tabular}
	\label{tab:StateAlfWorldm}
 \vspace{-1mm}
\end{table}

\begin{table}[!ht]\scriptsize
	\centering
	\caption{Environment (one category of questions) unlearning in HotPotQA with Different Number of Generated Unlearning Prompts ($m$).}
\begin{tabular} {ccc} 
\toprule
  & \makecell[c]{Unlearn efficacy\\($m=1/3/5$)} & \makecell[c]{Unlearn@1\\($m=1/3/5$)} \\
 \midrule
NL (GPT) & $0.88/1/1$ & $0.76/0.97/0.98$ \\ 
NL (Claude) & $0.86/1/1$ & $0.73/0.96/0.96$ \\
NL (Llama) & $0.91/1/1$ & $0.72/0.93/0.94$ \\
\bottomrule
\end{tabular}
	\label{tab:EnvironmentHotPotQAm}
 \vspace{-1mm}
\end{table}

\begin{table}[!ht]\scriptsize
	\centering
	\caption{Environment (one category of problems) unlearning in HumanEval with Different Number of Generated Unlearning Prompts ($m$).}
\begin{tabular} {ccc} 
\toprule
  & \makecell[c]{Unlearn efficacy\\($m=1/3/5$)} & \makecell[c]{Unlearn@1\\($m=1/3/5$)} \\
 \midrule
NL (GPT) & $0.85/1/1$ & $0.73/0.95/0.96$ \\ 
NL (Claude) & $0.82/1/1$ & $0.70/0.94/0.94$ \\
NL (Llama) & $0.84/1/1$ & $0.74/0.95/0.95$ \\
\bottomrule
\end{tabular}
	\label{tab:EnvironmentHumanEvalm}
 \vspace{-0mm}
\end{table}

\noindent\textbf{Summary.} Both the adoption of the conversion model and its fine-tuning are essential for achieving effective unlearning in our method. 
Moreover, the choice of base model also significantly influences performance. Models with sufficient representational capacity enable $\mathcal{C}$ to capture nuanced mappings between unlearning goals and actionable prompts. 
Also, the number of generated unlearning prompts ($m$) affects performance, where setting $m = 3$ provides an efficient balance.

\subsection{Scalability Study}
This study includes three aspects: the number of unlearning states, the number of unlearning trajectories, and the number of unlearning environments or task categories. The detailed results are provided in the Appendix. 

\vspace{-1mm}
\section{Simulation Study}
We use the Gazebo platform \cite{Gazebo} to simulate a household environment and deploy a mobile robot within it to perform various tasks, such as cleaning and aged-care assistance. We then conduct experiments under the three unlearning scenarios. The results are presented as heatmaps in Figures \ref{fig:Heatmap}.

\begin{figure}[ht]
\centering
	\begin{minipage}{1\textwidth}
    \subfigure[\scriptsize{Before unlearning}]{
    \includegraphics[scale=0.16]{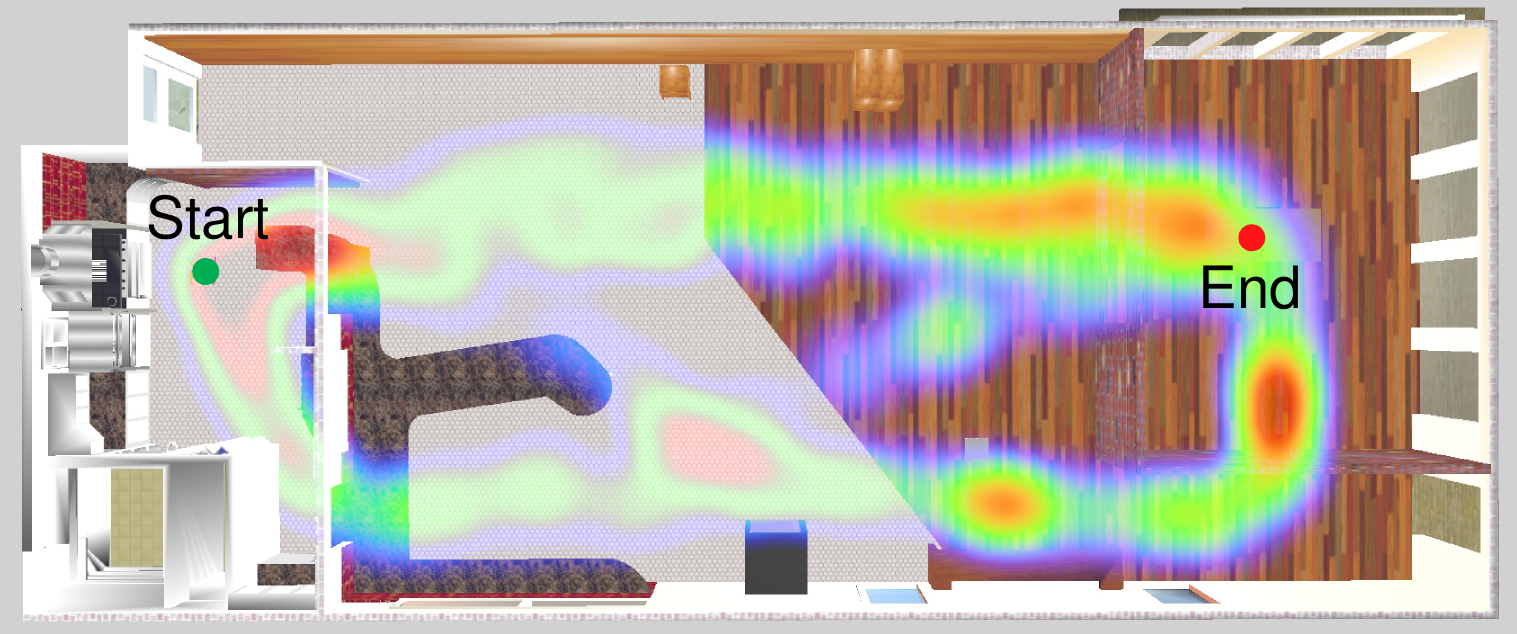}
			\label{fig:HeatmapBeforeUnlearn}}
    \subfigure[\scriptsize{State unlearning}]{
    \includegraphics[scale=0.16]{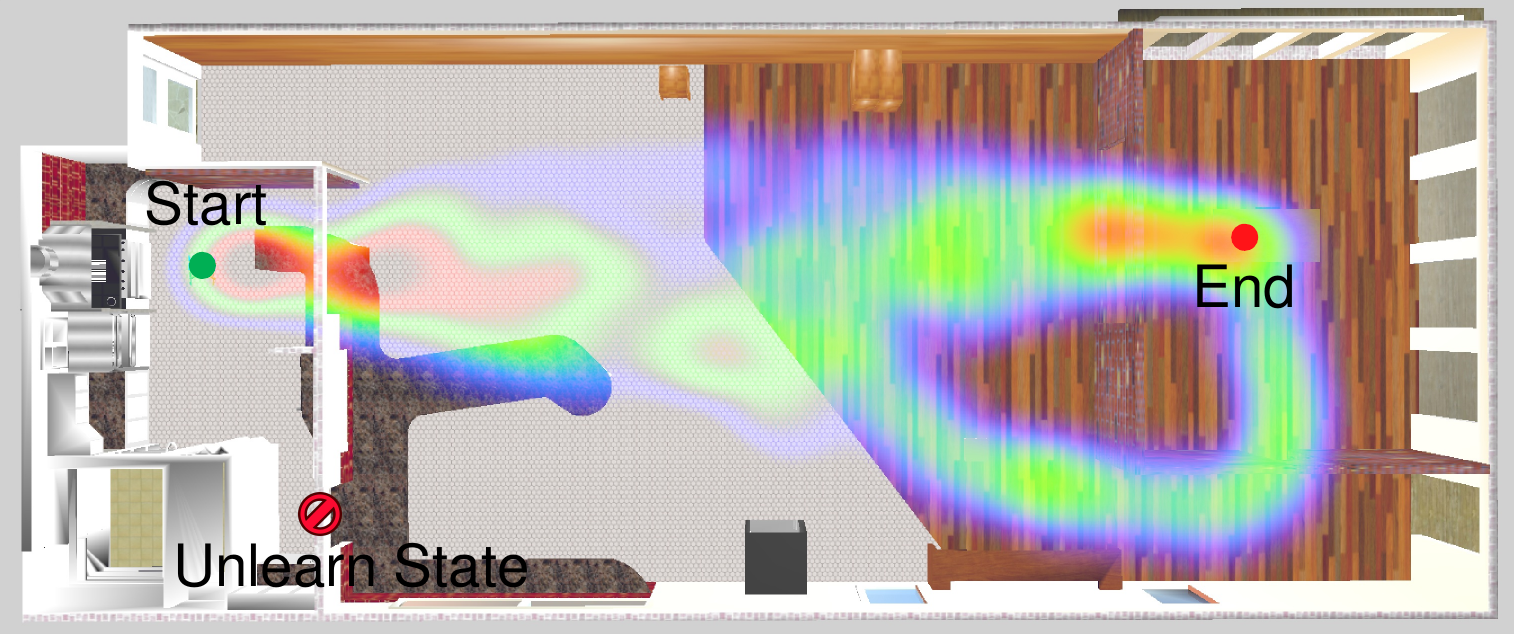}
			\label{fig:HeatmapStateUnlearn}}\\[0.5ex]
	\subfigure[\scriptsize{Trajectory unlearning}]{
    \includegraphics[scale=0.14]{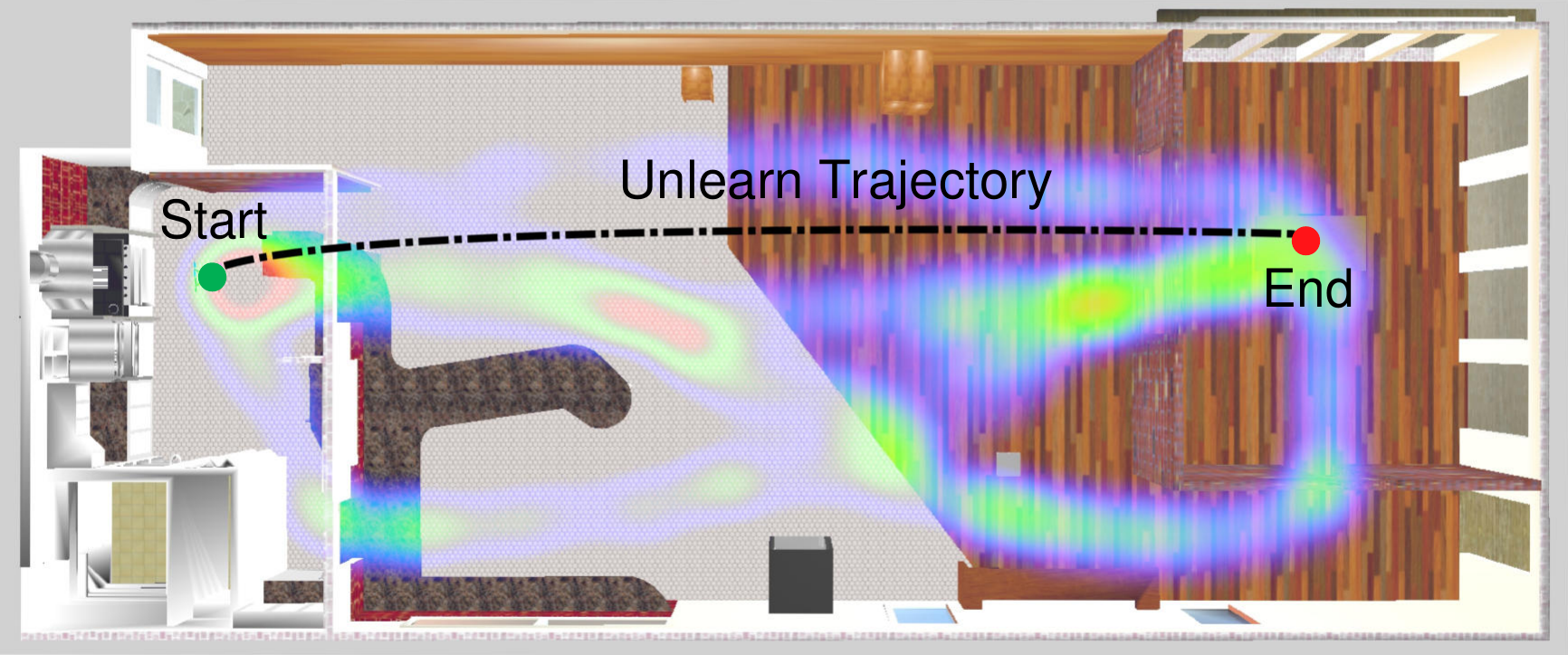}
			\label{fig:HeatmapTrajUnlearn}}
    \subfigure[\scriptsize{Environment unlearning}]{
    \includegraphics[scale=0.135]{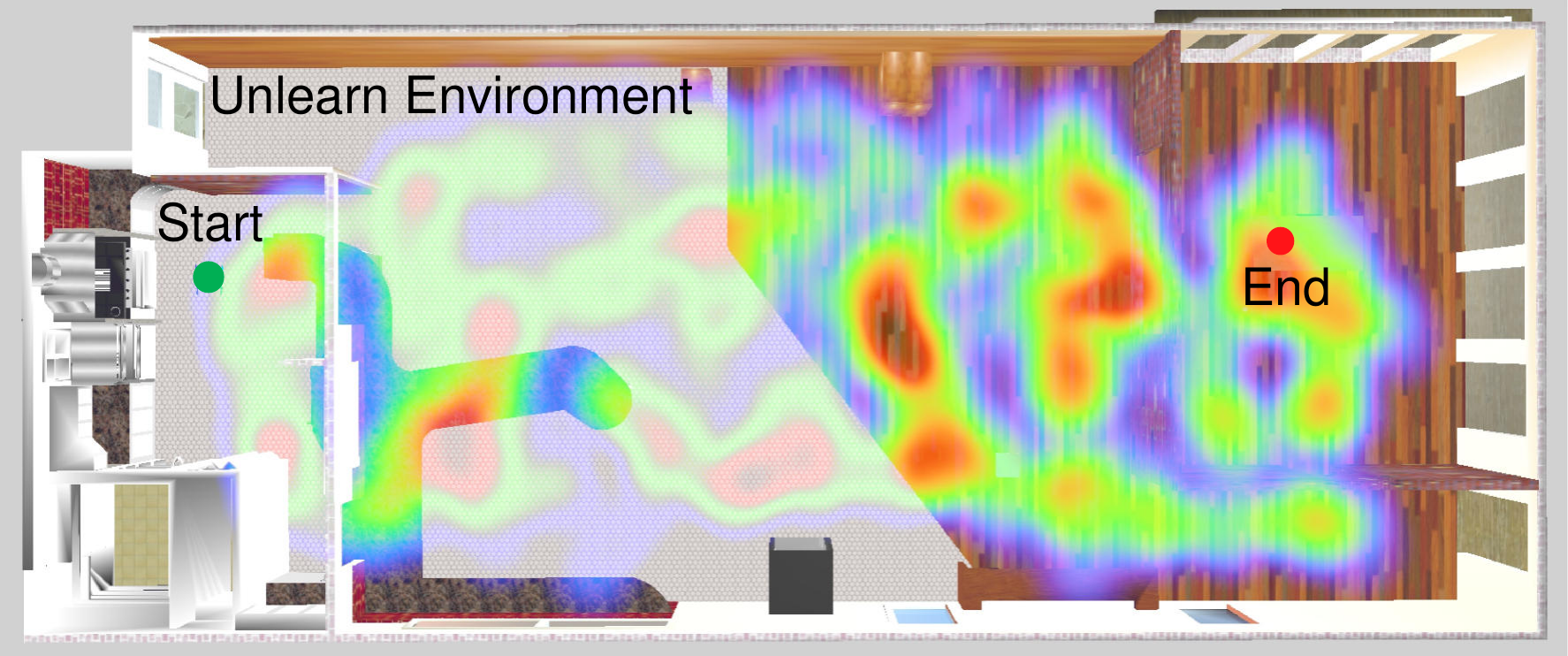}
			\label{fig:HeatmapEnvUnlearn}}\\[0ex]
    \end{minipage}
    \vspace{-2mm}
	\caption{Unlearning in the Gazebo platform}
    \vspace{-3mm}
	\label{fig:Heatmap}
\end{figure}

As shown in Figure \ref{fig:HeatmapBeforeUnlearn}, before unlearning, the robot traverses from the start point to the endpoint along two primary routes, where darker regions indicate higher visitation frequency. After performing state unlearning (Figure \ref{fig:HeatmapStateUnlearn}), in which the forgotten state corresponds to a door located at the bottom-left corner, the robot consistently avoids passing through that door. Similarly, after trajectory unlearning (Figure \ref{fig:HeatmapTrajUnlearn}), where the forgotten trajectory corresponds to the upper route, the robot refrains from using that path and instead discovers alternative routes. Finally, following environment unlearning (Figure \ref{fig:HeatmapEnvUnlearn}), the robot behaves randomly within the environment, as if exploring it for the first time. These results demonstrate that our method remains effective when deployed on realistic platforms.




\vspace{-1mm}
\section{Privacy Threats to Our Framework}\label{sec:robustness}

To evaluate the robustness of our unlearning framework, we design an unlearning inference adversary capable of interacting with the LLM-based agent through natural-language queries. The adversary issues task instructions, observes the agent’s behaviors, and attempts to infer the forgotten knowledge through behavioral analysis. A common attack against LLM-based agents is prompt injection \cite{Liu24USENIX}. However, prompt injection is not applicable in our setting, as it typically assumes a predefined target, such as steering the agent to select attacker-preferred tools \cite{Shi26NDSS}. In contrast, our adversary does not know which knowledge has been unlearned. 

Specifically, in state unlearning, the adversary aims to determine whether a given state $s$ has been unlearned. To do so, it constructs a set of tasks whose optimal solutions require visiting the suspected state. If the agent consistently selects suboptimal actions to avoid this state, the adversary gains strong confidence that the suspected state has been unlearned. Similarly, in trajectory unlearning, the adversary seeks to infer whether a given trajectory $\tau$ has been unlearned. It constructs tasks whose optimal solutions involve executing the suspected trajectory. If the agent persistently avoids forming this trajectory and instead adopts suboptimal alternatives, the adversary can confidently conclude that the suspected trajectory was unlearned. However, if the unlearned agent’s behavior closely resembles that of an agent that has never encountered the unlearned contents, even when the agent avoids them, the adversary cannot distinguish whether the contents were truly forgotten or never learned. 

In environment unlearning, the adversary’s objective shifts from identifying whether a specific environment has been unlearned to reconstructing the unlearned environment itself. This shift occurs because the adversary typically lacks access to the complete features and spatial layout of an environment, making binary inference infeasible. To achieve reconstruction, the adversary designs a set of exploratory tasks that prompt the agent to interact with different regions of the environment, and then applies a reconstruction method \cite{Pan19AAMAS} to infer its spatial layout based on the observed state–action trajectories.

The experiment is conducted in the GridWorld environment, chosen for its interpretability in observing agent behaviors. For implementation, in both state and trajectory unlearning scenarios, we define 10 pairs of starting and ending points located near the suspected state, encompassing the suspected trajectory, or traversing the target environment. Each pair of tasks is executed ten times. We then observe whether the agent visits the target state or trajectory during task execution, or, in the case of environment unlearning, reconstruct the target environment based on the agent’s behavioral patterns. In parallel, we construct a baseline agent that has never seen the unlearning contents by ensuring these contents are excluded from all prompt instructions during training and execution.

The results are presented in Tables~\ref{tab:StateGridWorldAttack},~\ref{tab:TrajGridWorldAttack}, and~\ref{tab:EnviGridWorldAttack}. As shown in Tables~\ref{tab:StateGridWorldAttack} and~\ref{tab:TrajGridWorldAttack}, when using our NL-based unlearning method, the probability that an unlearned agent traverses a forgotten state or trajectory is nearly identical to that of an agent that has never seen the forgotten state or trajectory. In contrast, the Code-based method performs slightly worse, as the unlearned agent exhibits a marginally higher probability of revisiting the forgotten state or trajectory. The Example-based method fails to achieve this effect due to its limited unlearning capability. 

\begin{table}[!ht]\scriptsize
\vspace{-1mm}
	\centering
	\caption{Inference Attack on State Unlearning in GridWorld.}
\begin{tabular} {cc} 
\toprule
  & \makecell[c]{Traverse the unlearned state\\(Pre-unlearn/Unlearn/Unseen)}  \\
 \midrule 
NL (GPT) & $0.99/0.02/0.01$  \\ 
NL (Claude) & $0.97/0.02/0.01$  \\ 
NL (Llama) & $0.98/0.01/0.02$  \\\hline 
Code (GPT) & $0.98/0.05/0.01$  \\ 
Code (Claude) & $0.98/0.08/0.02$  \\ 
Code (Llama) & $0.96/0.06/0.02$  \\\hline 
Exam. (GPT) & $0.95/0.24/0.01$   \\
Exam. (Claude) & $0.96/0.19/0.01$  \\ 
Exam. (Llama) & $0.93/0.18/0.02$   \\ 
\bottomrule
\end{tabular}
	\label{tab:StateGridWorldAttack}
 \vspace{-0mm}
\end{table}

\begin{table}[!ht]\scriptsize
\vspace{-2mm}
	\centering
	\caption{Inference Attack on Trajectory Unlearning in GridWorld.}
\begin{tabular} {cc} 
\toprule
  & \makecell[c]{Traverse the unlearned trajectory\\(Pre-unlearn/Unlearn/Unseen)}  \\ 
 \midrule
NL (GPT) & $0.95/0/0$  \\ 
NL (Claude) & $0.95/0.01/0.01$  \\ 
NL (Llama) & $0.94/0.03/0.02$  \\\hline 
Code (GPT) & $0.95/0.05/0$  \\ 
Code (Claude) & $0.94/0.07/0.02$  \\ 
Code (Llama) & $0.94/0.04/0.01$  \\\hline 
Exam. (GPT) & $0.96/0.27/0.01$   \\
Exam. (Claude) & $0.94/0.22/0.01$  \\ 
Exam. (Llama) & $0.95/0.23/0.02$   \\ 
\bottomrule
\end{tabular}
	\label{tab:TrajGridWorldAttack}
 \vspace{-2mm}
\end{table}

\begin{table}[!ht]\scriptsize
    \vspace{-1mm}
	\centering
	\caption{Reconstruction Attack on Environment Unlearning in GridWorld.}
\begin{tabular} {cc} 
\toprule
  & \makecell[c]{Successful reconstruction rate \\of the unlearned environment\\(Pre-unlearn/Unlearng/Unseen)}  \\ 
 \midrule
NL (GPT) & $0.96/0.18/0.15$  \\ 
NL (Claude) & $0.96/0.12/0.16$  \\ 
NL (Llama) & $0.93/0.15/0.13$  \\\hline 
Code (GPT) & $0.97/0.28/0.16$  \\ 
Code (Claude) & $0.95/0.33/0.14$  \\ 
Code (Llama) & $0.96/0.35/0.15$  \\\hline 
Exam. (GPT) & $0.98/0.43/0.17$   \\
Exam. (Claude) & $0.93/0.45/0.15$  \\ 
Exam. (Llama) & $0.92/0.40/0.16$   \\ 
\bottomrule
\end{tabular}
	\label{tab:EnviGridWorldAttack}
 \vspace{-0mm}
\end{table}

Similarly, Table~\ref{tab:EnviGridWorldAttack} shows that after applying our NL-based method for environment unlearning, the proportion of the unlearned environment that can be successfully reconstructed drops sharply to around $15\%$. This value is significantly lower than the pre-unlearning result (approximately $96\%$) and closely matches that of an agent that has never encountered the environment (about $14\%$). In contrast, the Code-based and Example-based methods yield considerably higher reconstruction proportions of around $33\%$ and $43\%$, respectively. 
Figure~\ref{fig:Reconstruction} illustrates an example of environment reconstruction results. Before unlearning, the reconstructed environment (Figure~\ref{fig:Reconstruction}(b)) is highly similar to the target unlearning environment (Figure~\ref{fig:Reconstruction}(a)), preserving most of its structural features. After applying our unlearning method with the GPT model (Figure~\ref{fig:Reconstruction}(c)), the similarity between the reconstructed and target environments decreases significantly, revealing minimal information about the original environment. The reconstructed environment derived from an agent that has never encountered the target environment (Figure~\ref{fig:Reconstruction}(d)) shows an analogously low level of similarity, confirming the effectiveness of our approach.
These findings demonstrate that our NL-based method effectively reshapes the agent’s action distribution to align with that of an agent trained without the forgotten data, making it hard for an adversary to infer the forgotten information through behavioral observation. 

\begin{figure}[ht]
\vspace{-0mm}
\centering
    \includegraphics[scale=0.4]{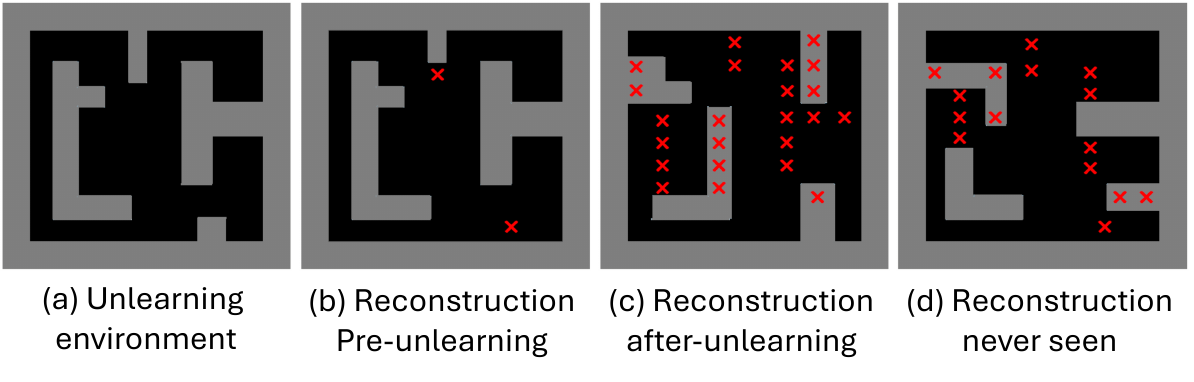}
    \vspace{-1mm}
	\caption{Environment reconstruction results. (a) The target environment to be unlearned. (b) The reconstructed environment before unlearning. (c) The reconstructed environment after unlearning using our method with GPT. (d) The reconstructed environment based on the behavior of an agent that has never seen the environment.}
    \vspace{-2mm}
	\label{fig:Reconstruction}
\end{figure}

\vspace{-1mm}
\section{Related Work}
\vspace{-0mm}
\noindent\textbf{LLM-based Agents.} Current LLM-based agents primarily focus on planning \cite{Sun23NIPS,Hu25KDD}, training \cite{Shinn23NIPS,Zheng24ICLR}, and reasoning and acting \cite{Yao23ICLR}. 
Sun et al. \cite{Sun23NIPS} introduce AdaPlanner, an LLM-based planner that refines agent plans based on feedback using both in-plan and out-of-plan refinement strategies, relying solely on prompting.
Zheng et al. \cite{Zheng24ICLR} propose SYNAPSE, an LLM-based agent controller that integrates three core components: state abstraction, which filters out task-irrelevant information from raw computer states; trajectory-as-exemplar prompting, which uses complete trajectories of abstracted states and actions to guide the LLM; and exemplar memory, which stores embeddings of successful trajectories for future retrieval.
Yao et al. \cite{Yao23ICLR} develop ReAct, a framework that enables LLMs to interleave reasoning (``thoughts'') and acting (``actions''), expanding the agent's action space by incorporating internal reasoning into decision-making.

\noindent\textbf{LLM Unlearning.} LLM unlearning methods can be broadly classified into weight-based and input-based approaches \cite{Liu24NatureMI}. Weight-based methods achieve unlearning by modifying the LLM's weights and/or architectural components \cite{Chen23EMNLP,Yao24ICLR}. In contrast, input-based methods rely on carefully designed input instructions, such as prompts or in-context examples, to guide the LLM toward the unlearning objective without requiring parameter updates \cite{Pawelczyk24ICML,Liu24NIPS,Sanyal25}. Yao et al. \cite{Yao24ICLR} propose an LLM unlearning method based on gradient ascent, which increases the gradient on the unlearning dataset, and random labeling, where responses are randomly mismatched with data in the unlearning dataset to disrupt learned associations. This line of research has been extended to address unlearning in pre-trained LLMs \cite{Yao24ACL}, continual unlearning over time \cite{Gao25ICLR}, and enhancing model safety through unlearning \cite{Liu24ACL}.
On the input-based side, Liu et al. \cite{Liu24NIPS} introduce an approach using embedding-corrupted prompts to address knowledge entanglement during unlearning. By modifying the embeddings of prompts, their method aligns the LLM's output with that of an LLM that has never seen the unlearning data, mitigating residual knowledge retention. This area has further evolved to explore evaluation metrics for quantifying forgetting \cite{Yuan25ICLR} and the dynamics between unlearning and re-learning \cite{Hu25ICLR}.


\vspace{-1mm}
\section{Conclusion} 
\vspace{-0mm}
This paper introduces a pioneering research area, termed LLM-based agent unlearning, which addresses critical concerns related to privacy and security. We proposed a novel unlearning framework that categorizes unlearning scenarios into three contexts. To realize this framework, we design an approach that trains a conversion model to translate high-level unlearning requests into actionable unlearning prompts, enabling precise forgetting within LLM-based agents. Future work will extend this framework to more interactive real-world scenarios, such as multi-agent systems. 

\bibliographystyle{IEEEtran}
\bibliography{references}

\section*{Appendix}
\setcounter{section}{0}
\renewcommand{\appendixname}{Appendix~\Alph{section}}

\section{Examples of Unlearning Requests and Prompts}
Table \ref{tab:Example} presents examples of unlearning requests and the corresponding unlearning prompts generated by the conversion model $\mathcal{C}$ in AlfWorld.

\begin{table}[h]\scriptsize
	\centering
    \vspace{-1mm}
 	\caption{Examples of Unlearning Requests and the Corresponding Unlearning Prompts in AlfWorld}
\begin{tabular} {p{1.1cm}|p{6.4cm}}
\toprule
\rowcolor{gray!20}
\multicolumn{2}{c}{\textbf{State Unlearning}}\\\cline{1-2}
\textbf{Unlearning Request} & \textit{You are an expert in designing prompts for household AI agents in ALFWorld environment. Your task is to put a cool mug in cabinet.... How should we modify the agent's prompt to make it completely avoid accessing `cabinet 6' during task execution? Please provide your response starting with ``Unlearning Instruction:'' followed by a detailed explanation.}\\\hline
\textbf{Unlearning Prompt} & \textit{Unlearning Instruction: Remove `cabinet 6' from the list of receptacles in the prompt. ..., the AI agent will no longer consider placing the cool mug in that specific cabinet during task execution. This modification ensures that the agent completely avoids accessing `cabinet 6' as it carries out the prompt, ...}\\\hline
\rowcolor{gray!20}
\multicolumn{2}{c}{\textbf{Trajectory Unlearning}}\\\cline{1-2}
\textbf{Unlearning Request} & \textit{...Task: Your task is to put a cool mug in cabinet. Available Options: [`cabinet 6', `cabinet 5', ..., `stoveburner 1', `toaster 1']. Goal: Avoid using that route by reversing the agent's search efficiency. Please provide your response starting with ``Unlearning Instruction:'' followed by a strategy description.}\\\hline
\textbf{Unlearning Prompt} & \textit{Unlearning Instruction: To reverse the agent's search efficiency, instruct the agent to prioritize further options first, rather than nearest ones. Specifically, the agent should begin its search with the farthest options (like 'garbagecan 1' or 'fridge 1') before exploring the more relevant cabinets...}\\\hline
\rowcolor{gray!20}
\multicolumn{2}{c}{\textbf{Environment Unlearning}}\\\cline{1-2}
\textbf{Unlearning Request} & \textit{You are designing environment-aware AI agents. Task: Your task is to: examine the alarmclock with the desklamp. Current Environment: [`bed 1', `desk 2', ..., `drawer 2', `drawer 1']. Question: How should the agent forget this specific environment? Please provide your response starting with ``Unlearning Instruction:'' followed by the forget strategy.}\\\hline
\textbf{Unlearning Prompt} & \textit{Unlearning Instruction: The agent should identify the specific environment pattern by checking for the presence of `bed 1', `desk 2', ..., `drawer 2', and `drawer 1' in the current environment. If these elements are present, the agent should skip the processing of the alarm clock and desk lamp, as this indicates the environment is the one to be forgotten...}\\\hline
\end{tabular}
\vspace{-0mm}
	\label{tab:Example}
\end{table}

\section{Proof of Key Conclusions}
\begin{proof}[Proof of Lemma \ref{lem:smoothness}]
By definition, $\mathcal{L}_{\mathcal{C}}$ is $L$-smooth if, for all $\phi_1, \phi_2$,
    \begin{equation}\nonumber
        ||\nabla_{\phi_1}\mathcal{L}_{\mathcal{C}}-\nabla_{\phi_2}\mathcal{L}_{\mathcal{C}}||\leq L||\phi_1-\phi_2||.
    \end{equation}
This holds if the operator norm of the Hessian satisfies
    \begin{equation}\nonumber
        ||\nabla^2_\phi\mathcal{L}_{\mathcal{C}}||_{\mathrm{op}}\leq L.
    \end{equation}
Using Eq. \ref{eq:second-order}, we obtain
    \begin{equation}\nonumber
        ||\nabla^2_\phi\mathcal{L}_{\mathcal{C}}||_{\mathrm{op}}=\beta^2\sigma(z)\sigma(-z)||\Delta_\psi||^2\leq\frac{\beta^2}{4}||\Delta_\psi||^2\leq\frac{\beta^2B^2}{4},
    \end{equation}
since $\sigma(z)\sigma(-z)\in(0,\frac{1}{4}]$. In particular, the boundedness assumption $|\Delta\psi|\leq B$ is mild and reasonable, as feature representations in practical models are typically normalized or regularized to maintain numerical stability.
\end{proof}

\begin{proof}[Proof of Lemma \ref{lem:convexity}]
A function $\mathcal{L}_{\mathcal{C}}$ is $\alpha$-strongly convex if, for all $\phi_1,\phi_2$,
    \begin{equation}\nonumber
        \mathcal{L}_{\mathcal{C}}(\phi_2)\geq\mathcal{L}_{\mathcal{C}}(\phi_1)+\nabla_{\phi_1}\mathcal{L}_{\mathcal{C}}^\top(\phi_2-\phi_1)+\frac{\alpha}{2}||\phi_2-\phi_1||^2.
    \end{equation}
which holds if and only if $\nabla^2_\phi\mathcal{L}_{\mathcal{C}}\succeq\alpha I$.
Substituting Eq. \ref{eq:second-order}, we have
    \begin{equation}\nonumber
        \nabla^2_\phi\mathcal{L}_{\mathcal{C}}=\beta^2\sigma(z)\sigma(-z)\Delta\psi(\Delta\psi)^\top\succeq\beta^2\epsilon\mu I,
    \end{equation}
thus $\alpha=\beta^2\epsilon\mu$. The assumption $\Delta\psi(\Delta\psi)^\top\succeq\mu I$ ensures sufficient feature diversity, which is typical in stochastic optimization and guarantees a well-conditioned curvature near the optimum.
\end{proof}

\begin{proof}[Proof of Theorem \ref{thm:convergence}]
    Gradient descent updates parameters as
    \begin{equation}\label{eq:SGD}
        \phi_{t+1}-\phi_t=-\eta\nabla_{\phi_t}\mathcal{L}_{\mathcal{C}}.
    \end{equation}
    Since $\mathcal{L}_{\mathcal{C}}$ is both $L$-smooth and $\alpha$-strongly convex, the descent lemma \cite{Bauschke17} gives
    \begin{equation}\label{eq:DescentLemma}
        \mathcal{L}_{\mathcal{C}}(\phi_{t+1})\leq\mathcal{L}_{\mathcal{C}}(\phi_{t})+\nabla_{\phi_t}\mathcal{L}_{\mathcal{C}}^\top(\phi_{t+1}-\phi_t)+\frac{L}{2}||\phi_{t+1}-\phi_t||^2.
    \end{equation}
Plugging Eq. \ref{eq:SGD} into (\ref{eq:DescentLemma}) yields: 
\begin{equation}\label{ieq:smooth}
\begin{aligned}
    \mathcal{L}_{\mathcal{C}}(\phi_{t+1})&\leq\mathcal{L}_{\mathcal{C}}(\phi_{t})-\eta||\nabla_{\phi_{t}}\mathcal{L}_{\mathcal{C}}||^2+\frac{L\eta^2}{2}||\nabla_{\phi_{t}}\mathcal{L}_{\mathcal{C}}||^2\\
    &=\mathcal{L}_{\mathcal{C}}(\phi_{t})-\eta(1-\frac{L\eta}{2})||\nabla_{\phi_{t}}\mathcal{L}_{\mathcal{C}}||^2.
\end{aligned}
\end{equation}

On the other hand, $\alpha$-strong convexity implies
\begin{equation}\label{ieq:convex}
\begin{aligned}
    &\mathcal{L}_{\mathcal{C}}(\phi^*)\geq\mathcal{L}_{\mathcal{C}}(\phi_t)+\nabla_{\phi_t}\mathcal{L}_{\mathcal{C}}^\top(\phi^*-\phi_t)+\frac{\alpha}{2}||\phi^*-\phi_t||^2\\
    \Rightarrow&\mathcal{L}_{\mathcal{C}}(\phi_t)-\mathcal{L}_{\mathcal{C}}(\phi^*)\leq\nabla_{\phi_t}\mathcal{L}_{\mathcal{C}}^\top(\phi_t-\phi^*)-\frac{\alpha}{2}||\phi_t-\phi^*||^2\\
    &\leq||\nabla_{\phi_t}\mathcal{L}_{\mathcal{C}}||\cdot||\phi_t-\phi^*||-\frac{\alpha}{2}||\phi_t-\phi^*||^2\\
    &\leq\frac{||\nabla_{\phi_t}\mathcal{L}_{\mathcal{C}}||^2}{\alpha}-\frac{||\nabla_{\phi_t}\mathcal{L}_{\mathcal{C}}||^2}{2\alpha}
    =\frac{||\nabla_{\phi_t}\mathcal{L}_{\mathcal{C}}||^2}{2\alpha}\\
    \Rightarrow&||\nabla_{\phi_t}\mathcal{L}_{\mathcal{C}}||^2\geq 2\alpha(\mathcal{L}_{\mathcal{C}}(\phi_t)-\mathcal{L}_{\mathcal{C}}(\phi^*)).
\end{aligned}
\end{equation}
The third inequality follows the Cauchy–Schwarz inequality, while the fourth arises from the fact that the concave quadratic function achieves its maximum at $||\phi_t-\phi^*||=\frac{||\nabla_{\phi_t}\mathcal{L}_{\mathcal{C}}||}{\alpha}$.
Finally, plugging (\ref{ieq:convex}) into (\ref{ieq:smooth}) achieves:
\begin{equation}\label{eq:conclusion}
\begin{aligned}
    &\mathcal{L}_{\mathcal{C}}(\phi_{t+1})
      \leq \mathcal{L}_{\mathcal{C}}(\phi_{t})
      -2\eta\alpha(1-\frac{L\eta}{2})
      (\mathcal{L}_{\mathcal{C}}(\phi_t)-\mathcal{L}_{\mathcal{C}}(\phi^*))\\[4pt]
    \Rightarrow&\mathcal{L}_{\mathcal{C}}(\phi_{t+1})-\mathcal{L}_{\mathcal{C}}(\phi^*)
      \leq
      (\mathcal{L}_{\mathcal{C}}(\phi_{t})-\mathcal{L}_{\mathcal{C}}(\phi^*))
      \\[-1pt]
      &\hspace{6em}
      -2\eta\alpha(1-\frac{L\eta}{2})
      (\mathcal{L}_{\mathcal{C}}(\phi_t)-\mathcal{L}_{\mathcal{C}}(\phi^*))\\[4pt]
    &
      \leq
      [1-2\eta\alpha(1-\frac{L\eta}{2})]
      (\mathcal{L}_{\mathcal{C}}(\phi_t)-\mathcal{L}_{\mathcal{C}}(\phi^*)).
\end{aligned}
\end{equation}

By choosing the step size $\eta\in(0,\frac{2}{L+\alpha}]$, the inequality in (\ref{eq:conclusion}) can be further simplified to:
\begin{equation}\nonumber
    \mathcal{L}_{\mathcal{C}}(\phi_{t+1})-\mathcal{L}_{\mathcal{C}}(\phi^*)\leq(1-\eta\alpha)(\mathcal{L}_{\mathcal{C}}(\phi_t)-\mathcal{L}_{\mathcal{C}}(\phi^*)),
\end{equation}
which concludes the proof.
\end{proof}

\begin{proof}[Proof of Theorem \ref{thm:KLBound}]
    When instantiating $Q(x,y)$ as $\beta r^*(x,y)$ in (\ref{eq:prop}), we obtain
\begin{equation}\label{eq:propOptimal}
    \pi^*_\phi(y|x)\propto\pi_{\mathrm{base}}(y|x)e^{\beta r^*(x,y)}.
\end{equation}
Similarly, its empirical counterpart can be expressed as
\begin{equation}\label{eq:propEmpirical}
    \pi_\phi(y|x)\propto\pi_{\mathrm{base}}(y|x)e^{\beta r(x,y)}.
\end{equation}
Dividing (\ref{eq:propEmpirical}) by (\ref{eq:propOptimal}) and taking the logarithm on both sides yields
\begin{equation}
\begin{aligned}
    \mathrm{log}\frac{\pi_\phi(y|x)}{\pi^*_\phi(y|x)}&=\beta(r(x,y)-r^*(x,y))\\
    \mathbb{E}_{y\sim\pi_\phi}[\mathrm{log}\frac{\pi_\phi(y|x)}{\pi^*_\phi(y|x)}]&=\beta\mathbb{E}_{y\sim\pi_\phi}[r(x,y)-r^*(x,y)]\\
    d_{\mathrm{KL}}(\pi_\phi(\cdot|x)\parallel \pi^*_\phi(\cdot|x))&\leq\frac{1}{\beta}\mathbb{E}_{y\sim\pi_\phi}[r^*(x,y)-r(x,y)],
\end{aligned}
\end{equation}
which indicates that the KL divergence between $\pi_\phi$ and $\pi^*_\phi$ is upper bounded by the expected reward difference between the optimal and empirical cases. For simplicity, let $\epsilon=\frac{1}{\beta}\mathbb{E}_{y\sim\pi_\phi}[r^*(x,y)-r(x,y)]$. We thus have
\begin{equation}\nonumber
    d_{\mathrm{KL}}(\pi_\phi(\cdot|x)\parallel \pi^*_\phi(\cdot|x))\leq\epsilon.
\end{equation}

Next, assume that the LLM-based agent policy is Lipschitz-continuous with respect to its prompt embedding, a common assumption in transformer-based models. Under this assumption, we have
\begin{equation}\label{eq:KLLip}
\begin{aligned}
    &d_{\mathrm{KL}}(\pi_\theta(\cdot|s_t;\pi_\phi(x))\parallel\pi_\theta(\cdot|s_t;\pi^*_\phi(x)))\\
    &\leq Ld_{\mathrm{KL}}(\pi_\phi(\cdot|x)\parallel \pi^*_\phi(\cdot|x)),
\end{aligned}
\end{equation}
where $L$ is the Lipschitz constant. Using the chain rule of KL divergence for sequential decision processes, we can write
\begin{equation}\label{eq:KLPunPref}
    d_{\mathrm{KL}}(P_{\mathrm{un}}\parallel P_{\mathrm{ref}})=\mathbb{E}_{s_t}[\sum_t d_{\mathrm{KL}}(\pi_\theta(\cdot|s_t;\pi_\phi(x_t))\parallel\pi_{\mathrm{ref}}(\cdot|s_t))].
\end{equation}
Since $\pi_\phi$ is trained to produce unlearning prompts that guide the LLM-based agent’s behavior, its optimal version $\pi^*_\phi$ can be regarded as generating prompts that perfectly achieve the desired unlearning objectives. Consequently, we can assume
\begin{equation}\nonumber
    \pi_\theta(a|s;\pi^*_\phi(x))\approx \pi_{\mathrm{ref}}(a|s).
\end{equation}
Substituting $\pi_{\mathrm{ref}}(\cdot|s_t)$ with $\pi_\theta(\cdot|s_t; \pi^*_\phi(x_t))$ in (\ref{eq:KLPunPref}) and applying (\ref{eq:KLLip}), we finally obtain 
\begin{equation}\nonumber
    d_{\mathrm{KL}}(P_{\mathrm{un}}\parallel P_{\mathrm{ref}})\leq L\mathbb{E}_{x}[d_{\mathrm{KL}}(\pi_\phi(\cdot|x)\parallel \pi^*_\phi(\cdot|x))]\leq L\epsilon.
\end{equation}
\end{proof}


\section{Additional Experimental Results}
\subsection{Scalability Study}


\noindent\textbf{Varying Numbers of Unlearning States.}
The results are shown in Tables \ref{tab:MultiStatesGridWorld}, \ref{tab:TwoStatesAlfWorld}, \ref{tab:MultiQuestionsHotPotQA}, and \ref{tab:MultiProblemsHumanEval} for the GridWorld, AlfWorld, HotPotQA, and HumanEval platforms, respectively. 

\begin{table}[!ht]\scriptsize
\vspace{-1mm}
	\centering
	\caption{Unlearning five/ten states in GridWorld.}
 \vspace{-0mm}
\begin{tabular} {ccccc} 
\toprule
  & \makecell[c]{Unlearn\\efficacy} & \makecell[c]{Steps\\Before\\unlearn} & \makecell[c]{Steps\\After unlearn\\(Target envir.)} & \makecell[c]{Steps\\After unlearn\\(Other envir.)} \\
 \midrule
NL (GPT) & $0.98$/$0.98$ & $20.1$ & $25.4$/$28.8$ & $20.0$/$20.2$ \\ 
NL (Claude) & $0.98$/$0.97$ & $21.3$ & $24.5$/$28.5$ & $21.4$/$21.3$ \\
NL (Llama) & $0.99$/$0.97$ & $19.8$ & $24.7$/$28.9$ & $19.9$/$20.0$ \\\hline
Code (GPT) & $0.97$/$0.96$ & $19.9$ & $25.6$/$29.7$ & $20.0$/$20.1$ \\
Code (Claude) & $0.99$/$0.97$ & $20.1$ & $25.5$/$29.4$ & $20.0$/$20.2$ \\
Code (Llama) & $0.99$/$0.97$ & $19.9$ & $25.2$/$29.6$ & $19.9$/$20.1$ \\\hline
Exam. (GPT) & $0.74$/$0.72$ & $20.2$ & $21.5$/$22.1$ & $20.3$/$20.3$ \\
Exam. (Claude) & $0.82$/$0.81$ & $19.9$ & $21.2$/$22.4$ & $19.9$/$20.0$ \\
Exam. (Llama) & $0.81$/$0.80$ & $20.1$ & $22.5$/$22.6$ & $20.1$/$20.2$ \\
\bottomrule
\end{tabular}
	\label{tab:MultiStatesGridWorld}
 \vspace{-1mm}
\end{table}


\begin{table}[!ht]\scriptsize
\vspace{-0mm}
	\centering
	\caption{Unlearning two states in AlfWorld.}
 \vspace{-0mm}
\begin{tabular} {ccccc} 
\toprule
  & \makecell[c]{Unlearn\\efficacy} & Unlearn@1 & \makecell[c]{Success rate\\Before unlearn} & \makecell[c]{Success rate\\After unlearn} \\
 \midrule
NL (GPT) & $1$ & $0.98$ & $0.97$ & $0.94$ \\ 
NL (Claude) & $1$ & $0.88$ & $0.86$ & $0.83$ \\
NL (Llama) & $1$ & $0.96$ & $0.95$ & $0.92$ \\\hline
Code (GPT) & $1$ & $0.96$ & $0.97$ & $0.82$ \\
Code (Claude) & $1$ & $0.82$ & $0.88$ & $0.78$ \\
Code (Llama) & $1$ & $0.88$ & $0.88$ & $0.60$ \\\hline
Exam. (GPT) & $1$ & $0.93$ & $0.20$ & $0.13$ \\
Exam. (Claude) & $1$ & $0.79$ & $0.07$ & $0.08$ \\
Exam. (Llama) & $1$ & $0.97$ & $0.40$ & $0.44$ \\
\bottomrule
\end{tabular}
	\label{tab:TwoStatesAlfWorld}
 \vspace{-1mm}
\end{table}

\begin{table}[!ht]\scriptsize
\vspace{-0mm}
	\centering
	\caption{Unlearning two/three questions in HotPotQA.}
 \vspace{-0mm}
\begin{tabular} {ccccc} 
\toprule
  & \makecell[c]{Unlearn\\efficacy} & Unlearn@1 & \makecell[c]{Success rate\\Before unlearn} & \makecell[c]{Success rate\\After unlearn} \\
 \midrule
NL (GPT) & $1$/$1$ & $1$/$1$ & $0.85$ & $0.80$/$0.82$ \\ 
NL (Claude) & $1$/$1$ & $0.5$/$0$ & $0.79$ & $0.76$/$0.74$ \\
NL (Llama) & $1$/$1$ & $1$/$1$ & $0.89$ & $0.85$/$0.84$ \\\hline
Code (GPT) & $1$/$1$ & $1$/$1$ & $0.81$ & $0.81$/$0.81$ \\
Code (Claude) & $1$/$1$ & $0.5$/$0$ & $0.76$ & $0.74$/$0.76$ \\
Code (Llama) & $1$/$1$ & $1$/$1$ & $0.83$ & $0.84$/$0.83$ \\\hline
Exam. (GPT) & $1$/$1$ & $1$/$1$ & $0.80$ & $0.81$/$0.80$ \\
Exam. (Claude) & $1$/$1$ & $0.5$/$0$ & $0.75$ & $0.75$/$0.75$ \\
Exam. (Llama) & $1$/$1$ & $1$/$1$ & $0.85$ & $0.86$/$0.86$ \\
\bottomrule
\end{tabular}
	\label{tab:MultiQuestionsHotPotQA}
 \vspace{-1mm}
\end{table}

\begin{table}[!ht]\scriptsize
\vspace{-0mm}
	\centering
	\caption{Unlearning two/three problems in HumanEval.}
 \vspace{-0mm}
\begin{tabular} {ccccc} 
\toprule
  & \makecell[c]{Unlearn\\efficacy} & Unlearn@1 & \makecell[c]{Success rate\\Before unlearn} & \makecell[c]{Success rate\\After unlearn} \\
 \midrule
NL (GPT) & $1$/$1$ & $1$/$1$ & $0.97$ & $0.95$/$0.93$ \\ 
NL (Claude) & $1$/$1$ & $0.5$/$0.33$ & $0.92$ & $0.88$/$0.90$ \\
NL (Llama) & $1$/$1$ & $0.5$/$0.33$ & $0.95$ & $0.89$/$0.89$ \\\hline
Code (GPT) & $1$/$1$ & $1$/$1$ & $0.94$ & $0.94$/$0.95$ \\
Code (Claude) & $1$/$1$ & $0.5$/$0.33$ & $0.89$ & $0.90$/$0.89$ \\
Code (Llama) & $1$/$1$ & $0.5$/$0.33$ & $0.90$ & $0.90$/$0.88$ \\\hline
Exam. (GPT) & $1$/$1$ & $1$/$1$ & $0.93$ & $0.95$/$0.94$ \\
Exam. (Claude) & $1$/$1$ & $0.5$/$0.33$ & $0.89$ & $0.90$/$0.90$ \\
Exam. (Llama) & $1$/$1$ & $0.5$/$0.33$ & $0.88$ & $0.88$/$0.90$ \\
\bottomrule
\end{tabular}
	\label{tab:MultiProblemsHumanEval}
 \vspace{-2mm}
\end{table}

The three unlearning methods demonstrate performance patterns similar to those observed in the scenario involving the unlearning of a single state, maintaining robust unlearning efficacy without severely impacting the execution of unrelated tasks.
However, an interesting observation arises in GridWorld (Table \ref{tab:MultiStatesGridWorld}). When employing the NL-based and Code-based methods, the average number of steps required to complete tasks in the target environment notably increases from approximately $20$ steps (before unlearning) to around $25$ and $29$ steps after unlearning $5$ and $10$ states, respectively. In contrast, the increase in the number of steps is substantially smaller for the Example-based method, reaching only about $21$ and $22$ steps after unlearning $5$ and $10$ states, respectively. 
This is because the NL-based and Code-based methods explicitly instruct the agent to strictly avoid the unlearned states, forcing the agent to select alternative, potentially suboptimal paths that lead to increased task execution steps. Conversely, the Example-based method implicitly guides the agent through provided examples rather than enforcing explicit constraints. Thus, the agent may retain partial access to previously learned knowledge, inadvertently utilizing parts of the original routes. 

\vspace{1mm}
\noindent\textbf{Varying Numbers of Unlearning Trajectories.}
This set of experiments applies exclusively to the GridWorld and AlfWorld settings, as the HotPotQA and HumanEval platforms do not incorporate the concept of trajectories. The results are presented in Tables \ref{tab:MultiTrajectoriesGridWorld} and \ref{tab:TwoTrajectoriesAlfWorld} for GridWorld and AlfWorld, respectively.
An interesting phenomenon emerges in the AlfWorld results (Table \ref{tab:TwoTrajectoriesAlfWorld}). Specifically, when using the Claude model, although the Unlearn@1 metric is relatively low, the task execution success rate after unlearning notably increases compared to before unlearning. 
This improvement can be attributed to Claude’s limited unlearning capability, meaning it is less aggressive and thus less prone to inadvertently forgetting useful task knowledge. Consequently, its conservative unlearning behavior effectively preserves important task execution information, resulting in an overall improvement in success rate after unlearning.

\begin{table}[!ht]\scriptsize
	\centering
	\caption{Unlearning two/three trajectories in GridWorld.}
\begin{tabular} {ccccc} 
\toprule
  & \makecell[c]{Unlearn\\efficacy} & \makecell[c]{Steps\\Before\\unlearn} & \makecell[c]{Steps\\After unlearn\\(Target envir.)} & \makecell[c]{Steps\\After unlearn\\(Other envir.)} \\
 \midrule
NL (GPT) & $0.98$/$0.97$ & $19.9$ & $22.8$/$23.1$ & $20.0$/$20.1$ \\ 
NL (Claude) & $0.98$/$0.96$ & $21.5$ & $22.8$/$23.2$ & $21.5$/$21.6$ \\
NL (Llama) & $0.98$/$0.97$ & $20.5$ & $22.3$/$23.1$ & $20.5$/$20.7$ \\\hline
Code (GPT) & $0.99$/$0.97$ & $20.2$ & $23.6$/$23.7$ & $20.2$/$20.1$ \\
Code (Claude) & $0.98$/$0.98$ & $20.9$ & $23.2$/$23.4$ & $21.0$/$21.0$ \\
Code (Llama) & $0.99$/$0.98$ & $20.3$ & $22.8$/$23.4$ & $20.4$/$20.3$ \\\hline
Exam. (GPT) & $0.81$/$0.80$ & $20.2$ & $21.8$/$21.8$ & $20.3$/$20.2$ \\
Exam. (Claude) & $0.81$/$0.79$ & $21.3$ & $22.7$/$22.7$ & $21.2$/$21.4$ \\
Exam. (Llama) & $0.83$/$0.83$ & $20.1$ & $21.9$/$22.2$ & $20.2$/$20.3$ \\
\bottomrule
\end{tabular}
	\label{tab:MultiTrajectoriesGridWorld}
\end{table}

\begin{table}[!ht]\scriptsize
	\centering
	\caption{Unlearning two trajectories in AlfWorld.}
\begin{tabular} {ccccc} 
\toprule
  & \makecell[c]{Unlearn\\efficacy} & Unlearn@1 & \makecell[c]{Success rate\\Before unlearn} & \makecell[c]{Success rate\\After unlearn} \\
 \midrule
NL (GPT) & $1$ & $0.98$ & $0.91$ & $0.88$ \\ 
NL (Claude) & $1$ & $0.81$ & $0.76$ & $0.82$ \\
NL (Llama) & $1$ & $1$ & $0.90$ & $0.87$ \\\hline
Code (GPT) & $1$ & $0.97$ & $0.90$ & $0.86$  \\
Code (Claude) & $1$ & $0.77$ & $0.74$ & $0.81$  \\
Code (Llama) & $1$ & $0.98$ & $0.88$ & $0.85$  \\\hline
Exam. (GPT) & $1$ & $0.58$ & $0.23$ & $0.19$ \\
Exam. (Claude) & $1$ & $0.35$ & $0.09$ & $0.14$ \\
Exam. (Llama) & $1$ & $0.52$ & $0.41$ & $0.40$ \\
\bottomrule
\end{tabular}
	\label{tab:TwoTrajectoriesAlfWorld}
\end{table}

\vspace{1mm}
\noindent\textbf{Varying Numbers of Unlearning Environments or Task Categories.}
The results are presented in Table \ref{tab:MultiEnvironmentGridWorld} (GridWorld), Table \ref{tab:TwoTasksAlfWorld} (unlearning two task categories in AlfWorld), Table \ref{tab:TwoEnvironmentsAlfWorld} (unlearning two environments in AlfWorld), Table \ref{tab:MultiCategoriesHotPotQA} (HotPotQA), and Table \ref{tab:MultiCategoriesHumanEval} (HumanEval).
An intriguing phenomenon emerges in the AlfWorld setting. As shown in Tables \ref{tab:TwoTasksAlfWorld} and \ref{tab:TwoEnvironmentsAlfWorld}, when the GPT model is instructed to unlearn two task categories, its success rate drops sharply after unlearning, indicating substantial performance degradation. A similar trend is observed when unlearning two environments, where both GPT and LLaMA exhibit notable reductions in success rates compared to their performance before unlearning.
These substantial reductions in performance can be attributed to the phenomenon of over-unlearning exhibited by GPT and Llama. Specifically, when asked to forget multiple categories or environments, these models tend to interpret the unlearning instructions too broadly. As a result, they overly restrict knowledge that extends beyond the intended targets. 

\begin{table}[!ht]\scriptsize
\vspace{-1mm}
	\centering
	\caption{Unlearning two/three environments in GridWorld.}
 \vspace{-0mm}
\begin{tabular} {ccccc} 
\toprule
  & \makecell[c]{Unlearn\\efficacy} & \makecell[c]{Steps\\Before\\unlearn} & \makecell[c]{Steps\\After unlearn\\(Target envir.)} & \makecell[c]{Steps\\After unlearn\\(Other envir.)} \\
 \midrule
NL (GPT) & $1$/$1$ & $20.5$ & $57.9$/$59.8$ & $21.2$/$21.2$ \\ 
NL (Claude) & $1$/$1$ & $21.7$ & $61.1$/$59.9$ & $22.2$/$21.1$ \\
NL (Llama) & $1$/$1$ & $20.7$ & $60.4$/$59.3$ & $22.3$/$21.7$ \\\hline
Code (GPT) & $1$/$1$ & $20.4$ & $59.7$/$61.2$ & $20.9$/$20.8$ \\
Code (Claude) & $1$/$1$ & $22.5$ & $58.7$/$60.5$ & $21.5$/$21.9$ \\
Code (Llama) & $1$/$1$ & $20.5$ & $59.1$/$58.6$ & $21.1$/$21.2$ \\\hline
Exam. (GPT) & $1$/$1$ & $20.4$ & $48.2$/$49.8$ & $21.2$/$22.1$ \\
Exam. (Claude) & $1$/$1$ & $21.6$ & $46.2$/$47.1$ & $21.4$/$22.8$ \\
Exam. (Llama) & $1$/$1$ & $20.6$ & $49.4$/$48.7$ & $21.3$/$20.9$ \\
\bottomrule
\end{tabular}
	\label{tab:MultiEnvironmentGridWorld}
 \vspace{-0mm}
\end{table}


\begin{table}[!ht]\scriptsize
\vspace{-0mm}
	\centering
	\caption{Unlearning two categories of tasks in AlfWorld.}
 \vspace{-0mm}
\begin{tabular} {ccccc} 
\toprule
  & \makecell[c]{Unlearn\\efficacy} & Unlearn@1 & \makecell[c]{Success rate\\Before unlearn} & \makecell[c]{Success rate\\After unlearn} \\
 \midrule
NL (GPT) & $1$ & $1$ & $0.85$ & $0.70$ \\ 
NL (Claude) & $1$ & $1$ & $0.73$ & $0.69$ \\
NL (Llama) & $1$ & $1$ & $0.95$ & $0.92$ \\\hline
Code (GPT) & $1$ & $1$ & $0.86$ & $0.62$ \\
Code (Claude) & $1$ & $0.75$ & $0.62$ & $0.68$ \\
Code (Llama) & $1$ & $1$ & $0.91$ & $0.90$ \\\hline
Exam. (GPT) & $1$ & $0.52$ & $0.33$ & $0.15$ \\
Exam. (Claude) & $1$ & $1$ & $0.13$ & $0.17$ \\
Exam. (Llama) & $1$ & $1$ & $0.38$ & $0.41$ \\
\bottomrule
\end{tabular}
	\label{tab:TwoTasksAlfWorld}
 \vspace{-0mm}
\end{table}

\begin{table}[!ht]\scriptsize
\vspace{-0mm}
	\centering
	\caption{Unlearning two environments in AlfWorld.}
 \vspace{-0mm}
\begin{tabular} {ccccc} 
\toprule
  & \makecell[c]{Unlearn\\efficacy} & Unlearn@1 & \makecell[c]{Success rate\\Before unlearn} & \makecell[c]{Success rate\\After unlearn} \\
 \midrule
NL (GPT) & $1$ & $1$ & $0.93$ & $0.84$ \\ 
NL (Claude) & $1$ & $0.98$ & $0.85$ & $0.82$\\
NL (Llama) & $1$ & $0.93$ & $0.91$ & $0.78$\\\hline
Code (GPT) & $1$ & $0.98$ & $0.95$ & $0.72$ \\
Code (Claude) & $1$ & $0.96$ & $0.89$ & $0.83$ \\
Code (Llama) & $1$ & $0.83$ & $0.88$ & $0.55$ \\\hline
Exam. (GPT) & $1$ & $0.94$ & $0.22$ & $0.11$ \\
Exam. (Claude) & $1$ & $0.92$ & $0.07$ & $0.06$ \\
Exam. (Llama) & $1$ & $0.85$ & $0.43$ & $0.34$ \\
\bottomrule
\end{tabular}
	\label{tab:TwoEnvironmentsAlfWorld}
 \vspace{-0mm}
\end{table}

\begin{table}[!ht]\scriptsize
\vspace{-0mm}
	\centering
	\caption{Unlearning two/three categories of questions in HotPotQA.}
 \vspace{-0mm}
\begin{tabular} {ccccc} 
\toprule
  & \makecell[c]{Unlearn\\efficacy} & Unlearn@1 & \makecell[c]{Success rate\\Before unlearn} & \makecell[c]{Success rate\\After unlearn} \\
 \midrule
NL (GPT) & $1$/$1$ & $1$/$1$ & $0.81$ & $0.82$/$0.79$ \\ 
NL (Claude) & $1$/$1$ & $0.65$/$0.69$ & $0.75$ & $0.80$/$0.79$ \\
NL (Llama) & $1$/$1$ & $0.88$/$0.91$ & $0.85$ & $0.96$/$0.80$ \\\hline
Code (GPT) & $1$/$1$ & $0.92$/$0.93$ & $0.81$ & $0.78$/$0.67$ \\
Code (Claude) & $1$/$1$ & $0.59$/$0.64$ & $0.75$ & $0.72$/$0.75$ \\
Code (Llama) & $1$/$1$ & $0.85$/$0.87$ & $0.85$ & $0.85$/$0.73$ \\\hline
Exam. (GPT) & $1$/$1$ & $0.95$/$0.96$ & $0.81$ & $0.77$/$0.77$ \\
Exam. (Claude) & $1$/$1$ & $0.71$/$0.66$ & $0.75$ & $0.76$/$0.76$ \\
Exam. (Llama) & $1$/$1$ & $0.87$/$0.88$ & $0.85$ & $0.86$/$0.82$ \\
\bottomrule
\end{tabular}
	\label{tab:MultiCategoriesHotPotQA}
 \vspace{-0mm}
\end{table}

\begin{table}[!ht]\scriptsize
\vspace{-0mm}
	\centering
	\caption{Unlearning two/three categories of problems in HumanEval.}
 \vspace{-0mm}
\begin{tabular} {ccccc} 
\toprule
  & \makecell[c]{Unlearn\\efficacy} & Unlearn@1 & \makecell[c]{Success rate\\Before unlearn} & \makecell[c]{Success rate\\After unlearn} \\
 \midrule
NL (GPT) & $1$/$1$ & $1$/$1$ & $0.94$ & $0.90$/$0.88$ \\ 
NL (Claude) & $1$/$1$ & $0.78$/$0.80$ & $0.89$ & $0.94$/$0.91$ \\
NL (Llama) & $1$/$1$ & $0.84$/$0.86$ & $0.89$ & $0.95$/$0.82$ \\\hline
Code (GPT) & $1$/$1$ & $1$/$1$ & $0.94$ & $0.94$/$0.91$ \\
Code (Claude) & $1$/$1$ & $0.76$/$0.77$ & $0.89$ & $0.94$/$0.97$ \\
Code (Llama) & $1$/$1$ & $0.81$/$0.83$ & $0.89$ & $0.92$/$0.94$ \\\hline
Exam. (GPT) & $1$/$1$ & $1$/$1$ & $0.94$ & $0.84$/$0.77$ \\
Exam. (Claude) & $1$/$1$ & $0.80$/$0.83$ & $0.89$ & $0.93$/$0.96$ \\
Exam. (Llama) & $1$/$1$ & $0.85$/$0.87$ & $0.89$ & $0.94$/$0.96$ \\
\bottomrule
\end{tabular}
	\label{tab:MultiCategoriesHumanEval}
 \vspace{-0mm}
\end{table}

\noindent\textbf{Summary.} The proposed NL-based method consistently outperforms the Code-based and Example-based methods in terms of unlearning efficacy when unlearning multiple items. However, after unlearning, the NL-based and Code-based methods typically require slightly more steps for task execution compared to the Example-based method in GridWorld. 
Further, although the Claude model exhibits comparatively lower unlearning performance than GPT and Llama, it effectively mitigates the risk of over-unlearning, particularly in scenarios involving multiple unlearning items.

\end{document}